\documentclass[11pt,reqno]{article}
\pdfoutput=1

\usepackage{amsmath,amsthm,amssymb}
\usepackage{setspace}

\usepackage[letterpaper,total={15.5cm, 22.4cm}]{geometry}

\usepackage{caption}
\usepackage{microtype}
\usepackage{tikz}
\usepackage{booktabs}

\numberwithin{equation}{section}
\newtheorem{theorem}{Theorem}[section]
\newtheorem{proposition}[theorem]{Proposition}
\newtheorem{lemma}[theorem]{Lemma}

\usepackage[USenglish]{babel}

\usepackage{cite}

\colorlet{darkblue}{blue!70!black}
\usepackage[colorlinks=true,urlcolor=darkblue,linktocpage=true,linkcolor=darkblue,citecolor=darkblue]{hyperref}

\numberwithin{equation}{section}

\newcommand{\R}{\mathbb{R}}
\newcommand{\Z}{\mathbb{Z}}
\newcommand{\btau}{\bar{\tau}}
\newcommand{\bq}{\bar{q}}
\newcommand{\bh}{\bar{h}}
\newcommand{\bx}{\bar{x}}
\newcommand{\vol}{\mathop{\textup{vol}}\nolimits}
\newcommand{\ML}{{\mathcal{M}_\ell}}
\newcommand{\p}{\partial}
\newcommand{\del}{\nabla}
\renewcommand{\Im}{\mathop{\textup{Im}}}

\begin{document}
\begin{spacing}{1.15}
\begin{titlepage}

\begin{center}
{\Large \bf
Free partition functions and an\\
averaged holographic duality}

\vspace*{6mm}

Nima Afkhami-Jeddi,$^1$ Henry Cohn,$^2$ Thomas Hartman,$^3$ and Amirhossein Tajdini$^3$

\vspace*{6mm}

$^1$\textit{Enrico Fermi Institute \& Kadanoff Center for Theoretical Physics,\\ University of Chicago, Chicago, Illinois, USA}

$^2$\textit{Microsoft Research New England, Cambridge, Massachusetts, USA \\}

$^3$\textit{Department of Physics, Cornell University, Ithaca, New York, USA\\}

\vspace{6mm}

{\tt nimaaj@uchicago.edu, cohn@microsoft.com, hartman@cornell.edu, at734@cornell.edu}

\vspace*{6mm}
\end{center}
\begin{abstract}
We study the torus partition functions of free bosonic CFTs in two
dimensions. Integrating over Narain moduli defines an ensemble-averaged
free CFT.  We calculate the averaged partition function and show that it
can be reinterpreted as a sum over topologies in three dimensions. This
result leads us to conjecture that an averaged free CFT in two dimensions
is holographically dual to an exotic theory of three-dimensional gravity
with $\textup{U}(1)^c \times \textup{U}(1)^c$ symmetry and a composite boundary graviton.
Additionally, for small central charge $c$, we obtain general constraints
on the spectral gap of free CFTs using the spinning modular bootstrap,
construct examples of Narain compactifications with a large gap, and find
an analytic bootstrap functional corresponding to a single self-dual boson.
\end{abstract}

\setcounter{tocdepth}{1}
\tableofcontents

\end{titlepage}
\end{spacing}

\begin{spacing}{1.15}
\addtocounter{page}{1}

\section{Introduction}

Among the simplest conformal field theories in two dimensions are those with
a $\textup{U}(1)^c_\textup{left} \times \textup{U}(1)_\textup{right}^c$ current algebra, where
$c$ is the central charge. These CFTs are theories of $c$ free bosons,
familiar from toroidal compactifications in string theory.

In this paper we will revisit an old problem: mapping the landscape of torus
partition functions for free CFTs. In the first part of the paper, we
undertake a systematic analysis of constraints on the spectrum using
techniques from the modular bootstrap
\cite{Rattazzi:2008pe,Hellerman:2009bu,Friedan:2013cba,Collier:2016cls}. In
\cite{HMR} it was shown that the modular bootstrap for free CFTs is related
to the sphere packing problem. However, this relation holds only for the
spinless version of the modular bootstrap, which in terms of the torus
modulus is restricted to $\tau = - \btau$. Here we will apply the full
modular bootstrap, with independent $\tau$ and $\btau$, which does not appear
to be related to sphere packing in general.

Instead, the full modular bootstrap for free CFTs is related to the geometric
problem of constructing Narain lattices with a large spectral gap, which is a
special case of sphere packing. A \emph{Narain lattice} is an even self-dual
lattice in $\mathbb{R}^{c,c}$, which famously defines a theory of $c$ compact
bosons \cite{Narain:1985jj}. The \emph{spectral gap} is the scaling dimension
$\Delta_1$ of the first nontrivial primary state in the CFT defined by this
lattice, and a Narain lattice is \emph{optimal} if it maximizes this gap
among all such lattices with a given central charge.

The modular bootstrap places an upper bound on the gap as a function of $c$.
We compute this bound numerically for $c \leq 15$, compare the bounds to
explicit Narain lattices, and discuss cases where the numerical bound is
saturated. We analytically solve the case $c=1$, where the optimal theory is
a self-dual boson, by exhibiting a suitable bootstrap functional. This is an
interesting example for the bootstrap because while some spinless bootstrap
problems are analytically tractable \cite{Mazac:2016qev}, there are
relatively few exact results with spin (see, however, recent progress in
\cite{Paulos:2019gtx,Mazac:2019shk}).

In the second part of the paper, we use methods of Siegel
\cite{MR10574,MR67930,MR67931,MR0271028} to study free boson partition
functions averaged over Narain moduli. These methods provide an
ensemble-averaged formula for the density of states in a free CFT, where the
ensemble is defined by the natural measure on the moduli space provided by
the Zamolodchikov metric (which in this case agrees with the Haar measure for
$\textup{O}(c,c)$ up to scaling).\footnote{An ensemble of symmetric orbifold CFTs based
on Siegel's technique of averaging over Narain lattices was considered by Moore in \cite{Moore:2015bba},
with a different holographic interpretation.} In particular,
the formula provides information about the spectrum
of an average Narain lattice in a large number of dimensions, and we use it
to prove that as $c \to \infty$, there are Narain lattices with $\Delta_1
\geq c/(2\pi e) + o(c)$.

This formula for $\Delta_1$ motivates the search for a holographic duality.
To explain why, let us first step back to review the status of holographic
duality for pure gravity in three dimensions, and the corresponding search
for a dual CFT. A holographic dual for pure 3d gravity would be a CFT with
Virasoro chiral algebra and $\Delta_1/c$ finite and nonzero in the limit as
$c \to \infty$. No such CFT has been found. Indeed, to find or exclude such a
theory is one of the primary motivations of the modular bootstrap program.
The interpretation of such a CFT, if it exists, is that the Virasoro
descendants of the vacuum are dual to Brown-Henneaux boundary gravitons in
AdS$_3$, and the primaries with dimension of order $c$ are dual to black
holes or other non-perturbative states.

In 2007, Maloney and Witten \cite{Maloney:2007ud} calculated the path
integral for 3d Einstein gravity with a torus boundary condition. It takes
the form
\begin{equation}\label{MW}
Z_\textup{MW}(\tau, \btau) = \sum_{\gamma \in \textup{SL}(2,\mathbb{Z})/\Gamma_\infty} \chi^\textup{Vir}_0(\gamma \tau) \bar{\chi}^\textup{Vir}_0(\gamma \btau),
\end{equation}
where $\chi_{ 0}^\textup{Vir}$ is the Virasoro vacuum character, $\gamma
\tau$ is an image of $\tau$ under $\textup{SL}(2,\mathbb{Z})$,  and the other notation
will be explained in Section~\ref{ss:bpf}. This sum over images under the modular group is
known as a Poincar\'e series.\footnote{For related applications of Poincar\'e
series in holography, see, for example,
\cite{Dijkgraaf:2000fq,Manschot:2007ha,Cheng:2011ay,Castro:2011zq,Jian:2019ubz}.} In the gravity theory,
it is a sum over topologies of the BTZ black hole. Maloney and Witten
computed the sum and found that the result does not make sense as a CFT,
because the density of states is continuous and non-unitary. There have been
various steps toward fixing the unitarity problem
\cite{Keller:2014xba,Benjamin:2019stq,Alday:2019vdr,Benjamin:2020mfz}, most
recently by including conical defects in the path integral, but the resulting
spectrum is still continuous and the status of pure 3d gravity as a quantum
theory is as yet unresolved.

Another wrinkle in this story is the recent discovery that pure gravity in
two dimensions, where it is known as Jackiw-Teitelboim (JT) gravity, is
holographically dual to random matrix theory
\cite{Cotler:2016fpe,Saad:2018bqo,Saad:2019lba,Saad:2019pqd}. This duality
provides a beautiful interpretation for a theory with a continuous spectrum
as an ensemble average over ordinary theories with discrete spectra. Since JT
gravity is the dimensional reduction of 3d gravity
\cite{Mertens:2018fds,Cotler:2018zff}, it seems increasingly likely that
averaging could also play a role in a putative dual to pure 3d gravity. On
the other hand, the notion of a random CFT in two dimensions is rather
mysterious: what is the ensemble? There is a natural measure on the moduli
space of CFTs connected by exactly marginal deformations, but a CFT dual to
pure 3d gravity would have no marginal operators. It would be isolated in the
space of CFTs. Therefore even if we had a large class of theories to average
over, it would be unclear how to define a measure.

We will show that if the Virasoro algebra is replaced by the $\textup{U}(1)^c$ current
algebra, then the sum over three-dimensional topologies can be carried out,
and it has a consistent interpretation as an average over Narain lattices. We
will refer to the bulk theory in three dimensions as \emph{$\textup{U}(1)$ gravity}.
It is perturbatively equivalent to $\textup{U}(1)^c \times \textup{U}(1)^c$ Chern-Simons
theory, with the action
\begin{equation}\label{sbulk}
S_\textup{CS}  =  \sum_{i=1}^c \int_{{\mathcal{M}}_3}( A^idA^i - \tilde{A}^id\tilde{A}^i) .
\end{equation}
We emphasize that this action is not supposed to define the non-perturbative
theory, and it is provisional in the sense that we will only check it on the
torus. For comparison, ordinary 3d gravity is perturbatively equivalent to an
$\textup{SL}(2,\mathbb{R}) \times \textup{SL}(2,\mathbb{R})$ Chern-Simons theory
\cite{Achucarro:1987vz,Witten:1988hc,Witten:2007kt}, with a boundary
condition inherited from gravity that differs from the usual one in gauge
theory (see, for example, \cite{Cotler:2018zff}). In addition to the
perturbative action \eqref{sbulk}, $\textup{U}(1)$ gravity comes with a prescription
to sum over three-dimensional topologies. This is part of the definition of
the theory. We will not attempt give a complete non-perturbative definition
in this paper, but for torus boundary conditions, the sum over topologies is
taken to be a sum over torus handlebodies, as in the Maloney-Witten path
integral \eqref{MW} for ordinary 3d gravity.

The theory of $\textup{U}(1)$ gravity is certainly not an ordinary gravitational
theory in three dimensions, so the lessons learned from this theory do not
necessarily carry over to more realistic theories. We do not expect it to
have black holes that dominate the canonical ensemble at $O(1)$ temperature.
On the other hand, $\textup{U}(1)$ gravity does have excitations equivalent to the
Brown-Henneaux boundary gravitons in ordinary 3d gravity. They are composites
built from the $\textup{U}(1)$ gauge fields, mimicking the Sugawara construction in
the boundary CFT. There are also higher spin composites, built from higher
products of the gauge fields, so $\textup{U}(1)$ gravity has some similarities to
higher spin gravity.\footnote{It differs from Vasiliev's theory of higher
spin gravity \cite{Bekaert:2005vh}, and there is no obvious relationship
between our results and previous examples of higher spin AdS/CFT
\cite{Sundborg:2000wp,Sezgin:2002rt,Klebanov:2002ja,Gaberdiel:2010pz}.}

The one-loop partition function for $\textup{U}(1)$ gravity on a solid torus is the
$\textup{U}(1)^c \times \textup{U}(1)^c$ vacuum character, denoted
$\chi_0(\tau)\bar{\chi}_0(\btau)$. Therefore the full partition function for
$\textup{U}(1)$ gravity on a torus is the Poincar\'e series
\begin{equation}
Z(\tau, \btau) = \sum_{\gamma \in \textup{SL}(2,\mathbb{Z})/\Gamma_\infty} \chi_0(\gamma \tau) \bar{\chi}_0(\gamma \btau) .
\end{equation}
We will compute the sum and show that the resulting spectrum agrees exactly
with Siegel's measure on random Narain lattices for any $c>2$. The agreement
between these two calculations is in fact a special instance of the
Siegel-Weil formula relating Eisenstein series to integrated theta functions
\cite{MR67930,MR67931,MR165033,MR223373}: the bulk calculation reduces to an
Eisenstein series, and the CFT calculation is an averaged theta function.

Thus, we conjecture that an averaged Narain CFT for $c>2$ is holographically
dual to a theory of $\textup{U}(1)$ gravity. We have demonstrated that this duality
holds at the level of the torus partition function, but we have not given a
fully non-perturbative definition of the bulk theory, which would require an
understanding of how to sum over topologies when the boundary condition is a
union of Riemann surfaces of arbitrary genus. If the duality is correct, then
it should also be possible to calculate ensemble-averaged quantities such as
$\langle Z(\tau_1,\btau_1) Z(\tau_2,\btau_2)\rangle$ from multi-boundary
wormholes in the bulk, as in the JT/random matrix duality
\cite{Saad:2019lba}. The connection to the Siegel-Weil formula also suggests
a way to generalize the calculations to higher genus.

Higher topology contributions to the gravitational path integral have played
a key role in recent efforts to address Hawking's information paradox
\cite{Penington:2019kki,Almheiri:2019qdq}. Whether these wormholes correspond
to an ensemble average is unknown, but in \cite{Marolf:2020xie}, it was
argued that spacetime wormholes in averaged theories can be reinterpreted by
doing the path integral with a boundary condition that selects an individual
member of the ensemble. It would be interesting to explore these alpha states
in $\textup{U}(1)$ gravity, where both sides of the duality are tractable.

In Section~\ref{s:prelim} we review background material on partition
functions with $\textup{U}(1)^c \times \textup{U}(1)^c$ symmetry. In Section~\ref{s:bootstrap},
we study bootstrap constraints and explicit Narain compactifications in low
dimensions. Finally, in Sections~\ref{s:averaging}--\ref{s:holography} we explore
averaging over Narain lattices and the holographic duality. The bootstrap
section is largely independent of the later sections, except as motivation,
so it can be read independently.

As this work was nearing completion, we learned that related ideas regarding
averaging over Narain lattices were arrived at independently by Maloney and
Witten \cite{Maloney:2020nni}.

\section{Preliminaries}\label{s:prelim}

\subsection{Partition functions}

The partition function of a compact, unitary 2d CFT is
\begin{equation}
Z(\tau, \btau) = \sum_\textup{states} q^{h - c/24} \bq^{\bh - c/24},
\end{equation}
where $q = e^{2\pi i \tau}$, $\bq  = e^{-2\pi i \btau}$, $h$ and $\bh$ are
non-negative conformal weights of each state, and $\tau$ and $-\btau$ are
independent complex numbers in the upper half-plane. In a theory with
$\textup{U}(1)^c_\textup{left} \times \textup{U}(1)^c_\textup{right}$ current algebra, the
partition function can be expressed as a sum over primaries via
\begin{equation}\label{zspin}
Z(\tau, \btau) = \sum_{h,\bh} d_{h,\bh} \chi_h(\tau) \bar{\chi}_{\bh}(\btau) ,
\end{equation}
where $\chi_h$ denotes the $\textup{U}(1)^c$ character
\begin{equation}
\chi_h(\tau) =\frac{q^h}{\eta(\tau)^c}  ,
\end{equation}
with $\eta$ the Dedekind eta function $\eta(\tau) = q^{1/24}
\prod_{n=1}^\infty(1-q^n)$ and $\bar{\chi}_{\bh}(\btau) =
\chi_{\bh}(-\btau)$, and the degeneracy $d_{h,\bh}$ is the number of
primaries with conformal weights $h$ and $\bh$. There is a unique vacuum
state with $h = \bh = 0$ and $d_{0,0}=1$.

We assume the partition function is modular invariant. In other words, $Z$
satisfies the identity
\begin{equation}
Z( \gamma \tau , \gamma \btau) = Z(\tau, \btau)
\end{equation}
for all $\gamma \in \textup{SL}(2,\Z)$, where $\gamma = \begin{pmatrix} p & q \\ r & s
\end{pmatrix}   \in \textup{SL}(2, \Z)$ acts as
\begin{equation}
(\gamma \tau, \gamma \btau) = \left( \frac{p \tau + q}{r \tau + s} , \ \frac{p \btau + q}{r \btau + s} \right) .
\end{equation}
The group $\textup{SL}(2,\Z)$ is generated by $S$ and $T$, where
\begin{equation}
S(\tau) = -1/\tau \qquad\text{and}\qquad T(\tau)=\tau+1.
\end{equation}
The scaling dimension and spin of a state are
\begin{equation}
\Delta = h+\bh \qquad\text{and}\qquad \ell = h-\bh,
\end{equation}
respectively. Invariance under $T$ requires that $\ell \in \Z$. Thus, we can
also write the partition function as
\begin{equation}\label{zld}
Z(\tau, \btau) = \sum_{\ell=-\infty}^{\infty} \int_{|\ell|}^\infty d\Delta\, \rho_\ell(\Delta) \chi_{\ell, \Delta}(\tau, \btau) ,
\end{equation}
where
\begin{equation}
\chi_{\ell, \Delta}(\tau, \btau) = \chi_{(\Delta + \ell)/2}(\tau) \bar{\chi}_{(\Delta -\ell)/2}(\btau)   .
\end{equation}
The density of states $\rho_\ell(\Delta)$ is a sum of delta functions with
positive integer coefficients, and the unitarity bound $h\geq 0$, $\bh \geq
0$ implies that $\rho_\ell(\Delta)$ has support only for $\Delta \geq |\ell|$.

\subsection{Spinning modular bootstrap}\label{ss:spinningbootstrap}

The modular bootstrap is a version of the conformal bootstrap applied to 2d
partition functions. Following
\cite{Rattazzi:2008pe,Hellerman:2009bu,Friedan:2013cba,Collier:2016cls}, we
write the condition $Z(\tau, \btau) - Z(-1/\tau, -1/\btau) = 0$ for
$S$-invariance as
\begin{equation}\label{crossphi}
\sum_{h,\bh} d_{h,\bh}\Phi_{h,\bh}(\tau, \btau) = 0 ,
\end{equation}
where we symmetrize $h$ and $\bh$ to obtain
\begin{equation}
\Phi_{h,\bh} = \chi_h(\tau)\bar{\chi}_{\bh}(\btau) + \bar{\chi}_h(\btau) \chi_{\bh}(\tau) - \chi_h(-1/\tau) \bar{\chi}_{\bh}(-1/\btau) - \bar{\chi}_h(-1/\btau)\chi_{\bh}(-1/\tau) .
\end{equation}
Suppose $\omega$ is a linear functional acting on functions of
$(\tau,\btau)$, such that
\begin{equation}\label{opos1}
\omega(\Phi_{0,0}) > 0
\end{equation}
and
\begin{equation}\label{opos2}
\omega(\Phi_{h,\bh}) \geq 0
\end{equation}
whenever $h\geq 0$, $\bh \geq 0$, $h - \bh \in \Z$, and $h + \bh \geq
\Delta_\textup{gap}$ for some constant $\Delta_\textup{gap}$. Then every CFT
must have a primary state with scaling dimension below $\Delta_\textup{gap}$,
because otherwise
\begin{equation}
\sum_{h,\bh} d_{h,\bh}\omega(\Phi_{h,\bh})  \ge d_{0,0} \omega(\Phi_{0,0}) > 0,
\end{equation}
which contradicts the crossing equation \eqref{crossphi}.

This method can be applied to any chiral algebra. Our focus is on theories
with $\textup{U}(1)^c \times \textup{U}(1)^c$ symmetry, for which the space of functionals can
be found by the usual logic with some minor adjustments. Under $S$, the
$\textup{U}(1)^c$ characters transform by a Fourier transform in $\R^c$: for $x \in
\R^c$,
\begin{equation}
\chi_{|x|^2/2}(-1/\tau) = \int_{\R^c} dk\, e^{-2\pi i k \cdot x} \chi_{|k|^2/2}( \tau) .
\end{equation}
Thus, $S$ acts on the product $\chi_h(\tau) \bar{\chi}_{\bh}(\btau)$ as a
Fourier transform in $\R^{2c}$ with the identifications $h =
\frac{1}{2}|x|^2$ and $\bh = \frac{1}{2}|\bx|^2$ for $(x,\bx) \in
\big(\R^{c}\big)^2 = \R^{2c}$. It follows that under these identifications,
the function $\omega(\Phi_{h,\bh})$ is always an eigenfunction of the Fourier
transform in $\R^{2c}$ with eigenvalue $-1$. Furthermore, every $-1$
eigenfunction that is invariant under exchanging $x$ and $\bx$ occurs as
$\omega(\Phi_{h,\bh})$ for some $\omega$, as one can check using the
derivative basis given in \eqref{deriv:basis} below.

In principle, the best bootstrap bound on $\Delta_1$ is obtained by
optimizing over this space of functionals. This is usually difficult, so it
becomes necessary to truncate the problem and use a computer to search over a
finite dimensional space. We restrict to the space spanned by the derivative
functionals
\begin{equation} \label{deriv:basis}
\omega =  \left.\frac{\partial^m }{\partial \tau^m} \frac{\partial^n}{\partial\btau^n}\right|_{\tau = - \btau = i }
\end{equation}
with $m + n \leq K$. The resulting eigenfunctions are spanned by
\begin{equation}
f_{m,n}(h, \bh) = \left(L_{m}^{(c/2-1)}(4\pi h)L_n^{(c/2-1)}(4\pi \bh)
+
L_{n}^{(c/2-1)}(4\pi h)L_m^{(c/2-1)}(4\pi \bh)\right)e^{-2\pi(h + \bh)}
\end{equation}
with $L_m^{(\nu)}(x)$ a generalized Laguerre polynomial, $m > n \geq 0$, $m +
n \leq K$, and $m+n$ odd. For each $(m,n)$, these functions have one discrete
label $h - \bh$, which we can take to be a non-negative integer, and one
continuous label $\Delta = h + \bh$.

\subsection{Narain compactifications}\label{ss:naraincomp}
A Narain lattice $\Lambda$ is an even, self-dual lattice of signature
$(c,\bar{c})$. For a review of the role of Narain lattices in conformal field
theory and string theory, see \cite{Polchinski:1998rq,Giveon:1994fu}. We consider Narain
lattices of signature $(c,c)$, where for $(x,y),(x',y') \in \R^{c,c} = (\R^c)^2$ the inner
product is $(x,y) \cdot (x', y') = x\cdot x' - y \cdot y'$. A Narain lattice
defines a CFT of $c$ free bosons, with the partition function
\begin{equation}
Z_\Lambda(\tau,\btau) = \frac{1}{\eta(\tau)^c \eta(-\btau)^c} \sum_{(x,y) \in \Lambda} q^{|x|^2/2} \bq^{|y|^2/2} .
\end{equation}
The condition that $\Lambda$ is even ensures that the CFT states have integer
spin, i.e., the partition function is invariant under $T$. Then the condition
that $\Lambda$ is self-dual implies that $Z_{\Lambda}$ is also invariant
under $S$ and therefore under the full modular group. The primary fields
correspond to vectors $(x,y) \in \Lambda$, with scaling dimension and spin
\begin{equation}
\Delta = \frac{1}{2}(|x|^2 + |y|^2)  \qquad\text{and}\qquad 
\ell = \frac{1}{2}(|x|^2 - |y|^2) .
\end{equation}
For each $c$, starting with a Narain lattice $\Lambda_0$, we can reach any
other Narain lattice by acting with an element of $\textup{O}(c,c)$  (see
\cite[Chapter~V]{Serre} or \cite[Chapter~II, \S5]{MH}). 
The CFT is invariant under the $T$-duality group $\textup{O}(\Lambda_0) \cong \textup{O}(c,c,\Z)$,
defined as the discrete subgroup of $\textup{O}(c,c,\R)$ which preserves the original lattice, and
the CFT is also unaffected by $\textup{O}(c)
\times \textup{O}(c)$ rotations acting individually on $x$ and $y$.
Therefore the moduli space of Narain CFTs is the quotient
\begin{equation}
\big( \textup{O}(c) \times \textup{O}(c) \big) \backslash \textup{O}(c,c) / \textup{O}(\Lambda_0).
\end{equation}
In the sigma model, this moduli space is parameterized by the metric and flux
on the target torus.

Consider $c=1$, the theory of a single compact boson of radius $R$. The
partition function is
\begin{equation}\label{zselfdual}
Z_R(\tau, \btau) = \frac{1}{\eta(\tau) \eta(-\btau)}\sum_{m,n \in \Z} q^{(m/R + nR/2)^2/2} \bq^{(m/R-nR/2)^2/2} .
\end{equation}
The theory is invariant under the $T$-duality $R \mapsto 2/R$. The spectrum
of primary operators is
\begin{equation}
\Delta_{m,n} = m^2/R^2 + n^2R^2/4 ,
\end{equation}
so the optimal Narain compactification for $c=1$ --- i.e., the CFT with the
largest gap between the vacuum state and the first nontrivial primary --- is
the self-dual boson, with $R = \sqrt{2}$. It has $\Delta_1 = 1/2$ for its
spectral gap.\footnote{At rational values of $R^2$ there are additional
conserved currents, so the chiral algebra is enhanced. 
The same happens at special values of the moduli for any
$c$. However we can still decompose states under $\textup{U}(1)^c \times \textup{U}(1)^c$, and
we will do this throughout the paper, so that the chiral characters are
independent of the moduli.}

More generally, the spectral gap of a Narain lattice $\Lambda$ is given by
\begin{equation}
\Delta_1 = \min_{(x,y) \in \Lambda \setminus\{(0,0)\}} \frac{|x|^2+|y|^2}{2}.
\end{equation}
In other words, we can form a sphere packing in ordinary Euclidean space
by centering spheres of radius
$\sqrt{2\Delta_1}$ at the points of $\Lambda$, with one sphere per unit
volume in space because all Narain lattices have determinant $1$. Maximizing
$\Delta_1$ amounts to maximizing the packing density. Thus, CFTs consisting
of free bosons correspond to a special case of the sphere packing problem, in
which the spheres must be centered at the points of a Narain lattice.

\section{Upper bounds on the spectral gap}\label{s:bootstrap}

\subsection{Numerical bootstrap bounds}
We use the spinning modular bootstrap method described in
Section~\ref{ss:spinningbootstrap}, together with standard computational
tools such as the semidefinite program solver SDPB
\cite{Simmons-Duffin:2015qma}, to place an upper bound on the spectral gap
$\Delta_1$ in theories with $\textup{U}(1)^c \times \textup{U}(1)^c$ current algebra. More
details of our implementation are in Appendix~\ref{appendix:sdpb}.

We denote the bootstrap bound at central charge $c$ and truncation order $K$
by $\Delta_{1}^{(K)}(c)$. All CFTs with this chiral algebra have $\Delta_1
\leq \Delta_1^{(K)}(c)$, and the bounds improve as $K \to \infty$. The
numerical results for $K=25$ are plotted in Figure~\ref{fig:spinplot}. The
red and green lines are included as a guide to the eye. To see the slight
nonlinearities in the bound, the piecewise linear function min$\left(
\frac{c+2}{6}, \frac{c+4}{8}\right)$ is subtracted from $\Delta_1^{(K)}(c)$
in Figure~\ref{fig:diffPW}. This figure also shows various values of $K$, so
that it can be used to judge whether the bound has converged. Some values
have converged better than others, and even some of the low-lying results may
not have converged. In particular the bounds around $c\sim 1.5$ and $c \sim
3$ are still changing appreciably at $K = 25$, so the actual bounds could be
significantly stronger. Note that larger values of the central charge require
a higher $K$ to get a strong bound, so it is not computationally feasible to
find useful bounds from this method for $c$ much larger than $15$. In
Figure~\ref{fig:diffS}, we compare to the spinless bootstrap bound for
$\textup{U}(1)^c \times \textup{U}(1)^c$ obtained in \cite{paper1}. For $c\neq 4$, the spinning
bound is strictly stronger in this range.

\begin{figure}
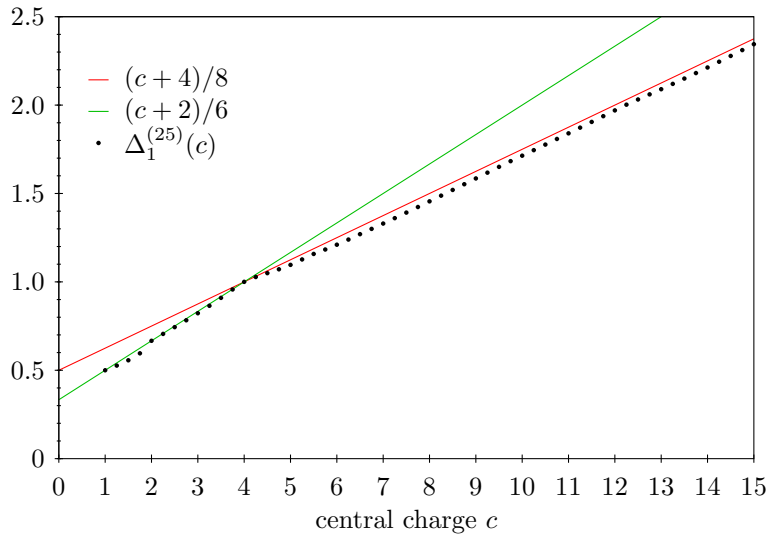

\begin{center}

\end{center}
\caption{Upper bound on $\Delta_1$ from the spinning modular bootstrap, at truncation order $K =25$. \label{fig:spinplot}}
\end{figure}

\begin{figure}
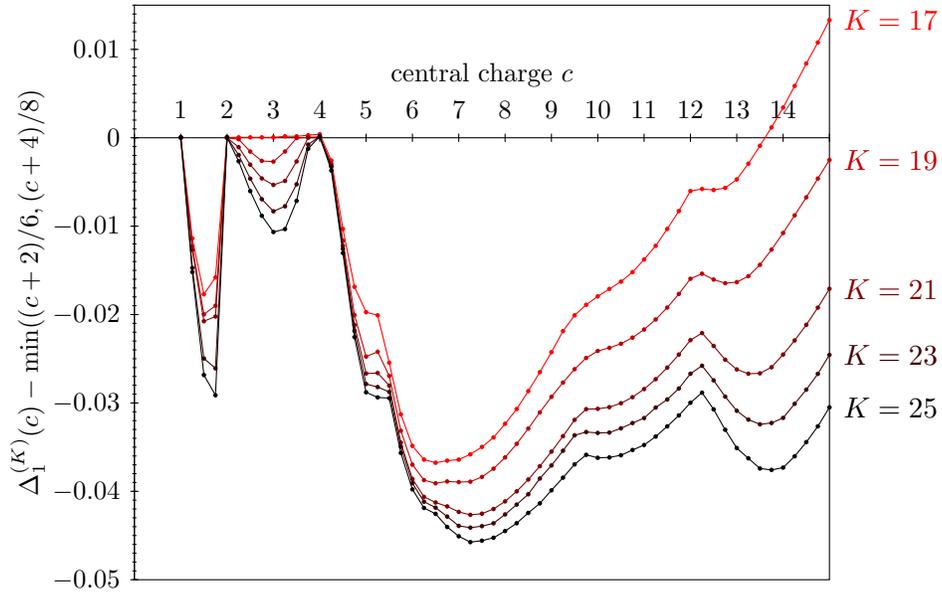

\begin{center}
%
\end{center}
\caption{Comparison of the upper bound to the piecewise linear function $\min\left(\frac{c+2}{6}, \frac{c+4}{8}\right)$.\label{fig:diffPW}}
\end{figure}

\begin{figure}
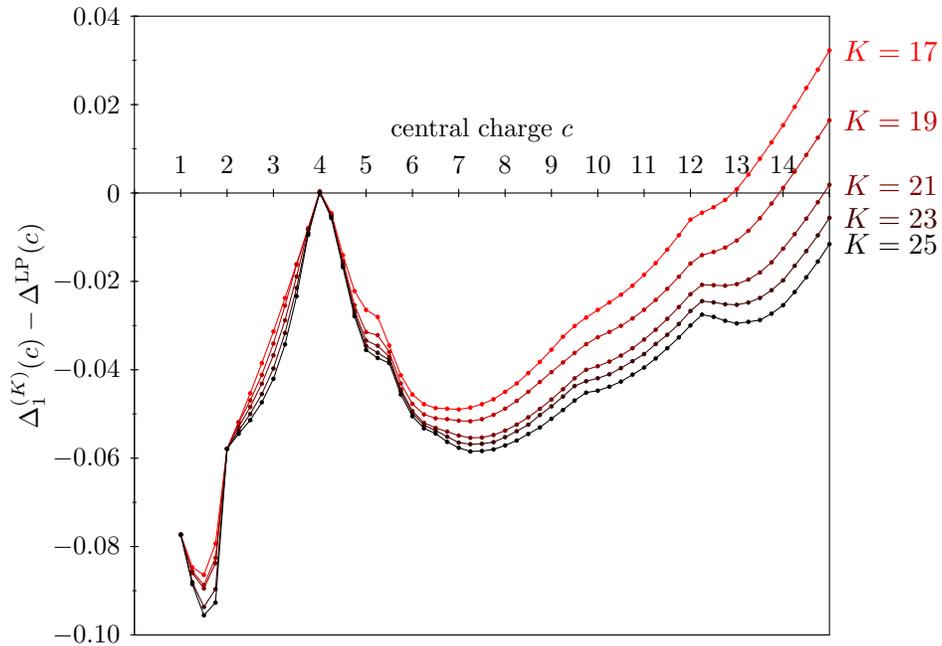

\begin{center}
%
\end{center}
\caption{Comparison of the spinning bootstrap to the spinless bootstrap bound $\Delta^\textup{LP}(c)$.\label{fig:diffS}}
\end{figure}

In Figure~\ref{fig:diffPW} we see that there are three points where the
spinning bound appears to converge to a known CFT, all sitting on the line
$\Delta = \frac{c+2}{6}$. The following upper bounds are obtained at
truncation order $K =19$:
\begin{equation}
\begin{split}
c&=1 : \qquad \Delta_1 < 1/2 + 2\times 10^{-51}\\
c&=2 :\qquad  \Delta_1 < 2/3 +  2\times 10^{-11} \\
c&=4 :\qquad \Delta_1 <   1 + 10^{-4}
\end{split}
\end{equation}
At $c=1$, the CFT that saturates the bound is a compact boson at the
self-dual radius, discussed in Section~\ref{ss:naraincomp}. This theory is
equivalent to the $\textup{SU}(2)_1$ WZW model. At $c=2$, the bound is saturated  by
the $\textup{SU}(3)_1$ WZW model. This theory has a realization as two bosons
compactified on a 2-torus at the three-fold symmetric point in moduli
space.  At $c=4$, as discussed in \cite{Collier:2016cls,paper1}, it is
saturated by 8 free fermions with the diagonal GSO projection, or
equivalently the $\textup{SO}(8)_1$ WZW model.

The sharp bound for $c=4$ follows automatically from the known bound using
the spinless modular bootstrap \cite{HMR}, and we will prove the bound for
$c=1$ below. That leaves the $c=2$ case as an open problem for the analytic bootstrap. It seems
conceptually similar to the sharpness of the spinless bound for $c=1$, and
both of these cases resist all known techniques.

The line $\Delta = \frac{c+2}{6}$ has appeared in previous modular bootstrap
studies \cite{Collier:2016cls,Bae:2017kcl}. It is the gap to the first
primary in the WZW models for
\begin{equation}
\textup{SU}(2)_1, \ \textup{SU}(3)_1,\  (G_2)_1,\  \textup{SO}(8)_1,\  (F_4)_1 , \ (E_6)_1 , \ (E_7)_1 ,
\end{equation}
with
\begin{equation}
c=1,\ 2,\  \frac{14}{5},\  4,\  \frac{26}{5},\  6,\  7,
\end{equation}
respectively. We have already encountered these theories at $c=1,2,4$. The
other theories on the list are consistent with our bound. Their partition
functions can be found in \cite{Mathur:1988na,Kiritsis:1988kq}. The $(G_2)_1$
WZW model does not have a $\textup{U}(1)^c_\textup{left} \times \textup{U}(1)^c_\textup{right}$
current algebra. When $c$ is not an integer, this algebra does not even make
sense, but we can still ask whether the partition function can be expanded as
in \eqref{zspin} with positive coefficients.  In the $(G_2)_1$ theory it
cannot, so the bound does not apply.  The $(E_6)_1$ and $(E_7)_1$ theories do
have the required current algebra. These theories have gap $\Delta_1 =
\frac{c+2}{6}$ with respect to the full chiral algebra, but gap $\Delta_1 =
1$ with respect to the  $\textup{U}(1)^c_\textup{left} \times \textup{U}(1)^c_\textup{right}$
subalgebra, because there are additional currents in the vacuum module that
are primary under this subalgebra.  Therefore they fall below our bound. The
situation for $(F_4)_1$ is similar. (This theory has no $\textup{U}(1)^c_\textup{left}
\times \textup{U}(1)^c_\textup{right}$ subalgebra, because $c$ is not an integer, but
does have a positive expansion of the form \eqref{zspin} with fractional
coefficients.)

\subsection{Analytic functional for the self-dual boson}

We will now construct an analytic functional to prove that every compact,
unitary 2d CFT with $c = \bar{c} = 1$ and current algebra $\textup{U}(1) \times \textup{U}(1)$
has a non-vacuum primary state with $\Delta_1 \leq \frac{1}{2}$. In other
words, the self-dual boson is optimal for this problem. This result may be
obvious, but the method is novel and may lend insight into more complicated
bootstrap problems with nontrivial spin dependence.

We can restate the requirements of Section~\ref{ss:spinningbootstrap} for the
spinning modular bootstrap in terms of Fourier eigenfunctions as follows. To
prove an upper bound of $\Delta_1 < \Delta_\textup{gap}$, we need a function
$f \colon \R^c \times \R^c \to \R$ such that $\widehat{f}=-f$, $f(0,0) > 0$,
and $f(x,\bx) \geq 0$ whenever $|x|^2 - |\bx|^2 \in 2\mathbb{Z}$ and $|x|^2 +
|\bx|^2 \geq 2\Delta_\textup{gap}$. For a rigorous proof, $f$ should decay
quickly enough; for example, a Schwartz function suffices. Without loss of
generality, we can assume that $f(x,\bx)$ depends only on $|x|^2$ and
$|\bx|^2$ and is invariant under exchanging $x$ and $\bx$.

The optimal choice of $f$ will have $f(0,0)=0$. We conjecture that replacing
the condition $f(0,0)>0$ with $f(0,0) \ge 0$ is enough to obtain $\Delta_1
\le \Delta_\textup{gap}$ as long as $f$ is not identically zero, but we do
not know how to prove it. We will first construct a function satisfying
$f(0,0)=0$ and $\Delta_\textup{gap}=\frac{1}{2}$ exactly, and then we will
approximate it with functions satisfying $f(0,0)>0$ and
$\Delta_\textup{gap}>\frac{1}{2}$ to obtain a rigorous proof.

To construct $f$, we begin with a convex subset $R$ of $\R^2$ that is
symmetric about the origin (in other words, $-R=R$). Let $\chi_R$ be the characteristic
function of $R$, i.e.,
\begin{equation}
\chi_R(x,\bx) = \begin{cases} 1 & \textup{if $(x,\bx) \in R$, and}\\
0 & \textup{otherwise},
\end{cases}
\end{equation}
and let $g = \chi_R * \chi_R$ be the convolution of $\chi_R$ with itself, so
that $g$ has support in $2R$. Then $\widehat{g} = \widehat{\chi_R}^2$, which
is nonnegative everywhere because $\widehat{\chi_R}$ is real-valued (which
holds since $R=-R$).  These functions satisfy $\widehat{g}(0,0) = \vol(R)^2$
and $g(0,0) = \vol(R)$, where here volume means area in $\R^2$.

Let $f = \widehat{g} - g$, so that $\widehat{f}=-f$. Then $f(0,0) \ge 0$ iff
$\vol(R) \ge 1$. We also want $f$ to satisfy $f(x,\bx) \ge 0$ whenever
$x^2+\bx^2 \ge 1$ and $x^2 - \bx^2 \in 2\Z$. We know that $f(x,\bx) \ge 0$
whenever $(x,\bx) \not\in 2R$, because $g$ vanishes outside $2R$ and
$\widehat{g}$ is always nonnegative. Thus it suffices to find $R$ such that
\begin{equation}
\{ (x,\bx) \in \R^2 : \text{$x^2+\bx^2 \ge 1$ and $x^2 - \bx^2 \in 2\Z$}\} \subseteq \R^2 \setminus 2R
\end{equation}
and $\vol(2R) \ge 4$. We can satisfy these conditions by taking $R$ to be a
square, namely the convex hull of $(\pm 1/\sqrt{2},0)$ and $(0,\pm
1/\sqrt{2})$, as shown in Figure~\ref{figure:square}. Thus, we have obtained
an optimal eigenfunction $f$, which in fact satisfies $f(x,\bx) \ge 0$ for
far more points $(x,\bx)$ than required.

\begin{figure}
\begin{center}
\begin{tikzpicture}[x=1.5cm,y=1.5cm,rotate=45]
\fill[black!15] (-1,1)--(1,1)--(1,-1)--(-1,-1)--(-1,1);
\draw (0,0) circle (1);
\draw (-1,1)--(1,1)--(1,-1)--(-1,-1)--(-1,1);
\draw [domain=0.5:2, samples=128]
 plot ({\x}, {1/\x} );
 \draw [domain=0.5:2, samples=128]
 plot ({-\x}, {1/\x} );
 \draw [domain=0.5:2, samples=128]
 plot ({\x}, {-1/\x} );
 \draw [domain=0.5:2, samples=128]
 plot ({-\x}, {-1/\x} );
\end{tikzpicture}
\end{center}
\caption{The square $2R$ that attains a sharp bound, together with the unit circle and the hyperbolas $x^2 - \bx^2  = \pm 2$.}
\label{figure:square}
\end{figure}
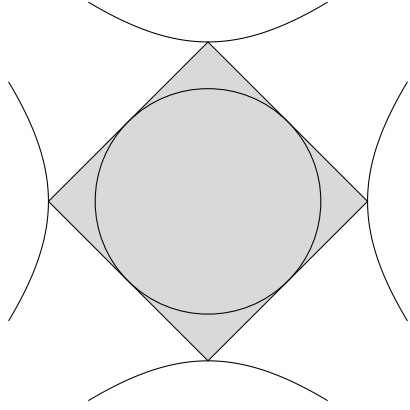

We can write down $g$ and $\widehat{g}$, and hence also $f$, in closed form
by rotating $45^\circ$ and separating variables, to take advantage of
decomposing the square $R$ as a product of two intervals. We find that
\begin{equation}
g(x,\bx) =  \left(1-\left|\frac{x+\bx}{\sqrt{2}}\right|\right)\chi_{[-1,1]}\mathopen{}\left(\frac{x+\bx}{\sqrt{2}}\right)\mathclose{}
\left(1-\left|\frac{x-\bx}{\sqrt{2}}\right|\right)\chi_{[-1,1]}\mathopen{}\left(\frac{x-\bx}{\sqrt{2}}\right)\mathclose{}
\end{equation}
and
\begin{equation}
\widehat{g}(x,\bx) =
\left(
\frac{\sqrt{2}\sin \frac{\pi(x + \bx)}{\sqrt{2} }}{\pi(x + \bx) }
\frac{\sqrt{2}\sin \frac{\pi(x - \bx)}{\sqrt{2} }}{\pi(x - \bx) }
\right)^2  .
\end{equation}
The function $f = \widehat{g}-g$ does indeed vanish at all the points
$(x,\bx) = \frac{1}{\sqrt{2}}(m+n, m-n)$ with $m,n \in \mathbb{Z}$, as it
should by \eqref{zselfdual}.

If we wish to achieve $f(0)>0$ while relaxing the constraint $x^2+\bx^2\ge 1$
to $x^2+\bx^2 \ge (1+\varepsilon)^2$ with $\varepsilon>0$, we cannot simply
replace $R$ with $(1+\varepsilon)R$, because the enlarged set
$2(1+\varepsilon)R$ would overlap with the hyperbolas $x^2=\bx^2 = \pm 2$.
Instead, we can shave off the corners of $(1+\varepsilon)R$ at $45^\circ$
angles to obtain on octagon $S_\varepsilon$ such that $2S_\varepsilon$
strictly avoids the hyperbolas, as shown in Figure~\ref{figure:octagon}. The
decrease in area from shaving the corners is quadratic in $\varepsilon$, and
thus $\vol(S_\varepsilon)>4$ when $\varepsilon$ is small. This construction
therefore comes arbitrarily close to $\Delta_\textup{gap}=\frac{1}{2}$ while
keeping $f(0,0)>0$.

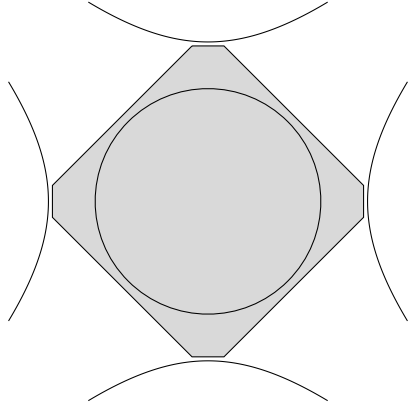
\begin{figure}
\begin{center}
\begin{tikzpicture}[x=1.5cm,y=1.5cm,rotate=45]
\fill[black!15] (0.875,1.075)--(1.075,0.875)--(1.075,-0.875)--(0.875,-1.075)--(-0.875,-1.075)--(-1.075,-0.875)--(-1.075,0.875)--(-0.875,1.075)--(0.875,1.075);
\draw (0,0) circle (1);
\draw (0.875,1.075)--(1.075,0.875)--(1.075,-0.875)--(0.875,-1.075)--(-0.875,-1.075)--(-1.075,-0.875)--(-1.075,0.875)--(-0.875,1.075)--(0.875,1.075);
\draw [domain=0.5:2, samples=128]
 plot ({\x}, {1/\x} );
 \draw [domain=0.5:2, samples=128]
 plot ({-\x}, {1/\x} );
 \draw [domain=0.5:2, samples=128]
 plot ({\x}, {-1/\x} );
 \draw [domain=0.5:2, samples=128]
 plot ({-\x}, {-1/\x} );
\end{tikzpicture}
\end{center}
\caption{An octagon $2S_\varepsilon$ such that $f(0,0)>0$, together with the unit circle and the hyperbolas $x^2 - \bx^2  = \pm 2$.}
\label{figure:octagon}
\end{figure}

The only remaining issue is that $f$ decays slowly. To fix this issue, we can
use a standard mollification argument, as in the proof of Lemma~2.2 in
\cite{CG}. Specifically, for each $\delta>0$ we can replace $f$ with a
Schwartz function $f_\delta$ such that $f_\delta$ converges pointwise to $f$
as $\delta \to 0$, $f_\delta(x,\bx) \ge 0$ whenever $(x,\bx) \not\in
(1+\delta) 2 S_\varepsilon$, and $\widehat{f_\delta}(x,\bx) \le 0$ whenever
$(x,\bx) \not\in (1+\delta) 2 S_\varepsilon$. Then the eigenfunction
$f_\delta - \widehat{f_\delta}$ has all the desired properties when $\delta$
is small enough.

Our construction of an optimal eigenfunction for $c=1$ is essentially
equivalent to the optimal auxiliary function for the one-dimensional sphere
packing bound from \cite[p.~695]{CE}: $g$ consists of two orthogonal copies
of the auxiliary function, at $45^\circ$ angles from the coordinate axes.
This relationship raises the question of whether the $c=2$ eigenfunction
might be related to an auxiliary function for $2$-dimensional sphere packing
in a similar way. It seems plausible that they are related somehow, but we
cannot pin down a specific relationship.

\subsection{Seeking optimal Narain lattices}

\begin{table}
\caption{Putatively optimal Narain compactifications, along with the spinning modular bootstrap
bound and the best lattice sphere packing known in $\R^{2c}$ (without the Narain condition).}
\label{table:narain}
\begin{center}
\begin{tabular}{cccccc}
\toprule
$c$ & $\Delta_1$ & $\Delta_{1}^{(25)}$ & Name & $\mkern-12mu$Best lattice packing\\
\midrule 
$1$ & $1/2$ & $1/2$ & $\textup{SU}(2)_1$ WZW$\phantom{1}$ & $\mkern-24mu\,\,\,\,\sqrt{1/3} = 0.5773\dotsc$\\
$2$ & $2/3$ & $0.6667$ & $\textup{SU}(3)_1$ WZW$\phantom{1}$ & $\mkern-24mu\,\,\,\,\sqrt{1/2} = 0.7071\dotsc$\\
$3$ & $3/4$ & $0.8227$ & $\textup{SU}(4)_1$ WZW$\phantom{1}$ & $\mkern-24mu\,\,\,\,\sqrt[6]{1/3} = 0.8326\dotsc$\\
$4$ & $1$ & $1$ & $\textup{SO}(8)_1$ WZW$\phantom{1}$ & $\mkern-24mu1\,\,\,\,\,\,$\\
$5$ & $1$ & $1.0963$ & $\textup{SO}(10)_1$ WZW $\phantom{1}$ & $\mkern-24mu\,\,\sqrt[10]{4/3} = 1.0291\dotsc$\\
$6$ & $\sqrt{4/3} = 1.1547\dotsc$ & $1.2103$ & Coxeter-Todd $\phantom{1}$ & $\mkern-24mu\,\,\,\,\,\sqrt{4/3} = 1.1547\dotsc$\\
$7$ & $\sqrt{4/3} = 1.1547\dotsc$ & $1.3300$ &  & $\mkern-24mu\sqrt[14]{64/3} = 1.2443\dotsc$\\
$8$ & $\,\,\,\,\,\,\,\sqrt{2} = 1.4142\dotsc$ & $1.4556$ & Barnes-Wall $\phantom{1}$ & $\mkern-24mu\,\,\,\,\,\,\,\,\,\,\,\,\,\sqrt{2} = 1.4142\dotsc$\\
\bottomrule
\end{tabular}
\end{center}
\end{table}

As discussed in Section~\ref{ss:naraincomp}, there is a unique Narain lattice
for each $c$, up to the action of the orthogonal group $\textup{O}(c,c)$. Therefore we
can try to find optimal Narain lattices by optimizing over moduli. This
optimization problem is highly non-convex, with many local optima. We
implemented a simple heuristic numerical algorithm, which starts from an
arbitrary element of $\textup{O}(c,c)$ and obtains a local optimum via hill climbing
under small, random perturbations. This algorithm does not perform well when
$c$ is large, but it gives good results for $c \le 8$. We used it to generate
a tentative list of optimal $(c,c)$ Narain compactifications, shown in
Table~\ref{table:narain}.

For $c\leq 5$ the best lattices we found are equivalent to WZW models at
level one. At $c=6$ or $8$, they turn out to be the Coxeter-Todd and
Barnes-Wall lattices, respectively. These lattices are scaled to have
irrational scaling dimensions, so they do not correspond to any WZW model.
They also happen to be the best sphere packings known in dimensions twelve or
sixteen, which means these cases cannot be improved without setting a new
record for the sphere packing density. The best Narain lattices match the
spinning modular bootstrap for $c=1$, $2$ (conjecturally), and $4$, but
seemingly not for $3$ or $5$ through $8$.  See Appendix~\ref{appendix:narain}
for further details and discussion of the Coxeter-Todd and Barnes-Wall
lattices.

\section{Averaging over Narain lattices}\label{s:averaging}

In this section, we review how Siegel computed the expected number of primary
states with specified scaling dimensions and spins in a random Narain CFT
\cite{MR67930}.\footnote{Of course, Siegel did not express his computation in
these terms in his 1951 paper.} Let $\Lambda_0$ be a Narain lattice in
$\R^{c,c}$, so that the space of all Narain lattices in $\R^{c,c}$ is the
orbit of $\Lambda_0$ under $\textup{O}(c,c)$, and let $\textup{O}(\Lambda_0)$ be the discrete subgroup
of $\textup{O}(c,c)$ that preserves $\Lambda_0$. Then the space of Narain lattices is
the quotient $\textup{O}(c,c)/\textup{O}(\Lambda_0)$.

The canonical measure on moduli spaces of CFTs is the Zamolodchikov metric
\cite{Zamolodchikov:1986gt}. For Narain CFTs, this measure is invariant under
$\textup{O}(c,c)$ and therefore agrees with the Haar measure on $\textup{O}(c,c)/\textup{O}(\Lambda_0)$,
up to scaling (see, for example, \cite{Moore:2015bba} for a detailed discussion). Thus, we can normalize to obtain a canonical probability measure
on Narain lattices if $\textup{O}(c,c)/\textup{O}(\Lambda_0)$ has finite volume under the Haar
measure. When $c=1$, the volume is infinite,\footnote{The metric on the
moduli space of a single free boson is proportional to $dR^2/R^2$, where $R$
is the target radius.} but it turns out to be finite for $c \ge 2$. This
finiteness can be checked directly by building a fundamental domain; it is
also a special case of the theorem of Borel and Harish-Chandra \cite{BHC}
that an arithmetic subgroup of a semisimple algebraic group has a finite
volume quotient (note that the identity component of $\textup{O}(c,c)$ is semisimple
iff $c \ge 2$). Thus, the notion of a uniformly random Narain lattice makes
sense for $c \ge 2$ but not $c=1$.

Narain lattices also behave unusually for $c=2$: the number of primary states
in a Narain CFT with spin $0$ and scaling dimension at most $\Delta$ grows
like a multiple of $\Delta^{c-1}$ as $\Delta \to \infty$ when $c>2$, but
there is an extra factor of $\log \Delta$ when $c=2$ (see Theorem~7 in
\cite{Heath-Brown}). In other words, Narain CFTs have excess spin $0$ states
when $c=2$, which leads to certain divergences. Siegel's theorem therefore
assumes $c>2$.

\begin{theorem}[Siegel] \label{theorem:siegel}
If $c>2$, then the density of non-vacuum primary states of spin $\ell$ and
scaling dimension $\Delta$ in a random Narain CFT of signature $(c,c)$ is
given by
\begin{equation}
\frac{2 \pi^c \sigma_{1-c}(\ell)}{\Gamma(c/2)^2 \zeta(c)} (\Delta^2-\ell^2)^{c/2-1}
\end{equation}
for $\Delta \ge |\ell|$ and $0$ otherwise. In other words, for each
measurable subset $A$ of $[|\ell|,\infty)$, the expected number of non-vacuum
primary states in a random Narain CFT with spin $\ell$ and scaling dimension
$\Delta \in A$ is
\begin{equation}
 \frac{2 \pi^c \sigma_{1-c}(\ell)}{\Gamma(c/2)^2 \zeta(c)} \int_A d\Delta \,  (\Delta^2-\ell^2)^{c/2-1}.
\end{equation}
\end{theorem}

Here $\sigma_{1-c}(\ell)$ is the sum of $m^{1-c}$ for all positive integers
$m$ dividing $\ell$, and we define $\sigma_{1-c}(0) := \zeta(c-1)$ since all
positive integers divide $0$. Note that $\zeta(c-1)$ is infinite if $c=2$,
and this divergence comes from the excess spin $0$ states.

Theorem~\ref{theorem:siegel} is implicit in \cite{MR67930}, and it is made
explicit in Theorem~8 in Chapter~4 of Siegel's TIFR lecture notes
\cite{MR0271028} (with somewhat cumbersome notation). In the rest of this
section, we will explain how one can compute these densities, while omitting
technicalities. First, we lay the groundwork by analyzing averaging over
Euclidean lattices.

\subsection{Averaging over lattices}

Before he proved Theorem~\ref{theorem:siegel}, Siegel dealt with the easier
case of Euclidean lattices of determinant $1$ in $\R^d$. The space of such
lattices is the orbit of $\Z^d$ under the action of $\textup{SL}(d,\R)$, i.e., the
quotient space $\textup{SL}(d,\R)/\textup{SL}(d,\Z)$. This homogenous space has finite volume
under the Haar measure for $\textup{SL}(d,\R)$, and thus we have a canonical
probability measure on the space of lattices. Siegel \cite{SiegelMVT} found
that the density of nonzero points in such lattices is $1$ if $d>1$. In other
words, for every measurable subset $A$ of $\R^d$, the expected number of
nonzero points in $A$ for a random lattice of determinant $1$ is $\vol(A)$.
(This assertion is clearly false for $d=1$, because there is a unique lattice
of determinant $1$ in $\R$, namely $\Z$.)

Setting aside technicalities, it is not hard to arrive at this answer. Let
$\mu$ be the measure on $\R^d$ for which $\mu(A)$ is the expected number of
lattice points in $A$. Then $\mu$ must be invariant under the action of
$\textup{SL}(d,\R)$ on $\R^d$. Because $\textup{SL}(d,\R)$ acts transitively on
$\R^d\setminus\{0\}$ for $d>1$ and preserves Lebesgue measure, the measure
$\mu$ must be of the form $\alpha \delta_0 + \beta \lambda$, where $\delta_0$
is a delta function at the origin, $\lambda$ is Lebesgue measure on $\R^d$,
and $\alpha, \beta \ge 0$, since the invariant measure on each orbit is
unique among regular measures. We must have $\alpha=1$, since the origin
occurs once in every lattice, and the only remaining question is what $\beta$
is. Because every lattice of determinant $1$ has $1$ point per unit volume on
a large enough scale, we conclude that $\beta=1$, as desired.

To make this argument rigorous, one must check several things. The most
important omissions are that the quotient $\textup{SL}(d,\R)/\textup{SL}(d,\Z)$ has finite
volume, that $\mu$ is a locally finite measure and in fact regular, and that
we have enough uniformity to justify the interchange of limits needed to
obtain the averaged assertion $\beta=1$ from facts about individual lattices.
All of these obstacles can be overcome; see \cite{SiegelMVT} or, for example,
\cite{MoskowitszSacksteder} or \cite{Venkatesh} for a modern perspective. We
will omit such issues below, and simply refer to \cite{MR67930} and
\cite{MR0271028} for a rigorous proof of Theorem~\ref{theorem:siegel}.

\subsection{Geometry of Narain lattices}

Let $\mu$ be the density measure for points in Narain lattices. In other
words, for $A \subseteq \R^{c,c}$, the expected number of points in $A$ for a
random Narain lattice is $\mu(A)$. As in the previous case, every orbit of
$\textup{O}(c,c)$ has a unique invariant measure, up to scaling, and the only question
is which scaling occurs for each orbit.

By definition, $\textup{O}(c,c)$ preserves the inner product $(x,y) \cdot (x',y') =
x\cdot x' - y \cdot y'$, if we represent elements of $\R^{c,c}$ as pairs of
vectors in $\R^c$ with the usual inner product in $\R^c$, and it acts
transitively on each hyperboloid
\begin{equation}
\{ (x,y) \in \R^{c,c} : |x|^2 - |y|^2 = t\}
\end{equation}
with $t \in \R$ except for $t=0$, in which case $\{(0,0)\}$ and $\{(x,y) :
|x|^2-|y|^2 = 0\} \setminus\{(0,0)\}$ are separate orbits (see Section~2 of
\cite{Strichartz}).

The orbit $\{(0,0)\}$ contributes a delta function, since the origin occurs
once in each Narain lattice. The other orbits are parameterized by $t =
2\ell$ for spin $\ell \in \Z$, and it is not difficult to write down the
invariant measures on these orbits (see, for example, Section~2 in
\cite{Strichartz}). Because the space of Narain lattices is invariant under
$\textup{O}(c) \times \textup{O}(c)$, all the information in these measures is contained in the
distribution of spins and scaling dimensions, i.e., a measure on $\Z \times
[0,\infty)$. We can compute this measure as follows. The homogeneous metric
on the hyperboloid $|x|^2-|y|^2=2\ell$ is proportional to
\begin{equation}\label{hypermetric}
-d\alpha^2  - \sinh^2 \alpha \, d\tilde{\Omega}_{c-1}^2  + \cosh^2 \alpha \, d\Omega_{c-1}^2,
\end{equation}
with $|x| = \sqrt{2\ell}\cosh\alpha$, $|y| = \sqrt{2\ell}\sinh\alpha$, and
$d\tilde{\Omega}_{c-1}^2$, $d\Omega_{c-1}^2$ each a line element on a unit
$(c-1)$-sphere. This formula is derived by parameterizing the hyperboloid as
$x= z \sqrt{2\ell} \cosh \alpha$, $y = \tilde{z} \sqrt{2\ell}\sinh\alpha$,
with $z$ and $\tilde{z}$ each unit vectors in $\mathbb{R}^{c}$, and plugging
into the line element $|dx|^2 - |dy|^2$. The corresponding volume element on the
hyperboloid is proportional to $(|x| |y|)^{c-1}d\alpha$. Because $|x|^2=\Delta+\ell$ and $|y|^2 = \Delta-\ell$, the density
of scaling dimensions $\Delta$ for spin $\ell$ is proportional to
$(\Delta^2-\ell^2)^{c/2-1}$ for $\Delta \ge |\ell|$, and of course it
vanishes otherwise, since no state can have $\Delta < |\ell|$.

Thus, the subtle content of Theorem~\ref{theorem:siegel} is the constants
\begin{equation}
\frac{2 \pi^c \sigma_{1-c}(\ell)}{\Gamma(c/2)^2 \zeta(c)}
\end{equation}
used to scale these measures, while the general form follows from the
$\textup{O}(c,c)$ symmetry. In the Euclidean case, there was only one missing
constant, which was easily determined, but here we must obtain infinitely
many constants. Fortunately, the same sort of argument works: every Narain
CFT with $c>2$ has the same asymptotic number of primary states of fixed spin
$\ell$ and scaling dimension at most $\Delta$, namely
\begin{equation} \label{Narain-asymp}
\left(\frac{2 \pi^c \sigma_{1-c}(\ell)}{\Gamma(c/2)^2 \zeta(c)} +o(1)\right) \frac{\Delta^{c-1}}{c-1}
\end{equation}
such states as $\Delta \to \infty$, which agrees with the Siegel density. All
that remains is to explain this formula.

\subsection{Counting states} \label{ss:counting}

To obtain the missing constants, we need to count states in a Narain CFT.
A closely related counting problem was treated in \cite{Moore:2015bba}.\footnote{One
of the main results of \cite{Moore:2015bba} is the volume of moduli space for symmetric product
CFTs with $N$ copies of a seed CFT. The calculation and final result are essentially the same
as in this subsection and appendices, with the replacement $\ell \to N$. The result of
\cite{Moore:2015bba} was interpreted as evidence that CFTs with a weakly coupled
holographic dual are rare. Our ensemble and our interpretation are different, but
not in disagreement with this conclusion since our bulk theory is not standard 3d gravity.}
To simplify the analysis, we choose null coordinates so that our quadratic form
of signature $(c,c)$ is given by $Q(x,y) = 2(x \cdot y)$ for $(x,y) \in
(\R^c)^2$. Then $(\Z^{c})^2$ is a Narain lattice (see
Appendix~\ref{appendix:narain}), and we will focus on this specific lattice
before generalizing to all Narain lattices.

The question is how many vectors in $(x,y) \in (\Z^{c})^2$ have $x \cdot y =
\ell$ and $|x|^2+|y|^2 \le r^2$ as $r \to \infty$. The Hardy-Littlewood
circle method gives an answer when $c>2$: the number of such vectors is
asymptotic to
\begin{equation}
\sigma_\infty(B_r) \prod_{\text{$p$ prime}} \sigma_p,
\end{equation}
where $B_r$ is the ball $\{(x,y) \in (\R^c)^2 : |x|^2+|y|^2 \le r^2\}$ of
radius $r$, $\sigma_\infty$ is the \emph{singular integral} defined by
\begin{equation}
\sigma_\infty(A) = \int_\R dt \int_{(x,y) \in A} dx\, dy\, e^{2\pi i ((x \cdot y)-\ell) t}
\end{equation}
for $A \subseteq (\R^c)^2$, and $\sigma_p$ is defined by
\begin{equation}
\sigma_p = \lim_{n \to \infty} \frac{\#\{(x,y) \in \big((\Z/p^n\Z)^{c}\big)^2 : x\cdot y \equiv  \ell \pmod{p^n}\}}{p^{(2c-1)n}}.
\end{equation}
The product $\prod_p \sigma_p$ is called the \emph{singular series}. The
intuition here is that we are counting integral solutions to the equation
$x\cdot y = \ell$, and each factor measures a different constraint:
$\sigma_\infty$ measures how many real solutions there are, and $\sigma_p$
measures how many solutions there are modulo high powers of $p$. There is no
reason to expect such an elegant answer in general, but it works here (see,
for example, \cite{MR1474964} or \cite{Heath-Brown} for the circle method, or
\cite{MR974910,MR1230289,MR1131433,MR1230290} for other approaches to these
sorts of counting problems). We will give a high-level description of the
method here, with some additional details in Appendix~\ref{app:circle}.

We begin by writing the lattice point count as a Fourier integral, namely
\begin{equation}
\#\{(x,y) \in B_r \cap (\Z^c)^2 : x \cdot y = \ell\} =
\int_0^1 dw\, \sum_{(x,y) \in B_r \cap (\Z^c)^2} e^{2\pi i (x \cdot y - \ell) w}.
\end{equation}
We would like to approximate this integral for large $r$, which requires
understanding where the integrand is large.

\begin{figure}
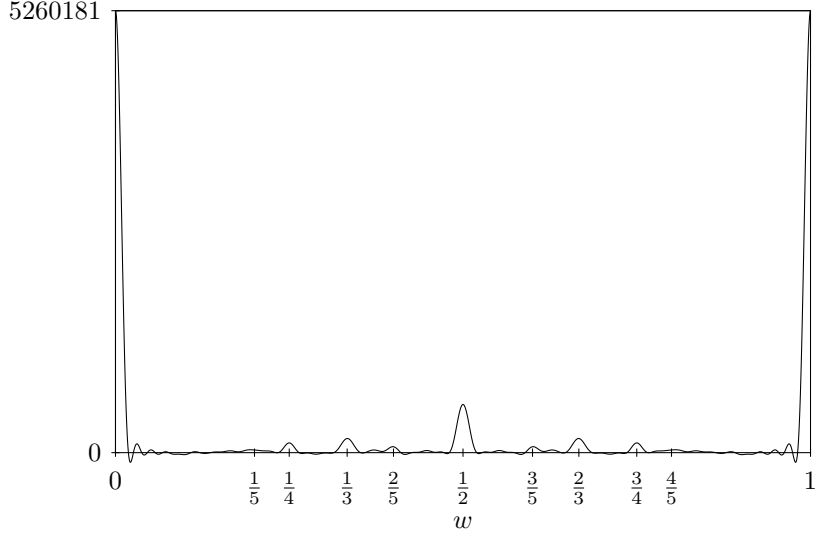

\begin{center}

\end{center}
\caption{The circle method integrand with $c=3$, $r=10$, and $\ell=0$.}
\label{figure:circle}
\end{figure}

The integrand is largest when $w=0$, in which case it simply counts the
lattice points in $B_r$ without regard for whether $x \cdot y = \ell$. It
turns out that the dominant contributions to the integral come from
intervals around rational numbers with small denominators, as illustrated in
Figure~\ref{figure:circle}. We will omit the estimates needed to prove this
assertion, as well as to bound the error terms throughout the argument;
instead, we will outline the calculations without fully justifying them.
Asymptotically, the dominant contributions come from $w$ in the \emph{major
arcs}
\begin{equation}
\left\{ w \in [0,1] : \left|w-\frac{a}{b}\right| \le \frac{1}{r^{2-\varepsilon}}\right\}
\end{equation}
for rational numbers $a/b$ in lowest terms with $1 \le b \le r^\varepsilon$, where $0<\varepsilon \ll 1$
(strictly speaking, we should wrap around and consider $w$ modulo $1$ to deal
with the endpoints), and the remaining \emph{minor arcs} turn out to
contribute a negligible amount.\footnote{The ``arc'' terminology comes from
integrating around the unit circle.} Note that the major arcs do not overlap,
and thus we can treat them independently. The remaining calculations amount
to approximating the integral over each major arc by an exponential sum times
the singular integral, and then factoring the sum of the resulting terms to
obtain the singular series. See Appendix~\ref{app:circle} for more details.

It is not hard to derive a recurrence for $\#\{(x,y) \in
\big((\Z/p^n\Z)^{c}\big)^2 : x\cdot y \equiv  \ell \pmod{p^n}\}$ (see
Appendix~\ref{app:counting}), and we find that
\begin{equation}
\sigma_p = \frac{(1-p^{-c})(1-p^{-(c-1)(k+1)})}{1-p^{-(c-1)}}
\end{equation}
if $p^k$ is the largest power of $p$ dividing $\ell$, where if $\ell=0$, we
take $k=\infty$ and therefore $p^{-(c-1)(k+1)}=0$.  A little manipulation
then shows that
\begin{equation}
\prod_{\text{$p$ prime}} \sigma_p = \frac{\sigma_{1-c}(\ell)}{\zeta(c)}.
\end{equation}
Furthermore, rescaling $(x,y)$ and $t$ shows that
\begin{equation}
\sigma_\infty(B_r) \sim r^{2c-2} \int_\R dt \int_{(x,y) \in B_1} dx\, dy\, e^{2\pi i (x \cdot y) t},
\end{equation}
and one can compute that
\begin{equation}
\int_\R dt \int_{(x,y) \in B_1} dx\, dy\, e^{2\pi i (x \cdot y) t} = \frac{\pi^c}{(c-1) 2^{c-2} \Gamma(c/2)^2}.
\end{equation}
Setting $\Delta = r^2/2$ yields the desired asymptotics for the Narain
lattice $(\Z^{c})^2$.

All that remains is to generalize this calculation to other Narain lattices.
Using the Iwasawa decomposition for $\textup{O}(c,c)$, we can reduce to the case of
lattices
\begin{equation}
\{ (A x + M (A^t)^{-1} y, (A^t)^{-1}y) : (x,y) \in (\Z^{c})^2\},
\end{equation}
where $A \in GL(c,\R)$ and $M \in \R^{c \times c}$ is antisymmetric (see
Proposition~\ref{prop:narainequiv} and the discussion following it). We have
\begin{equation}
(A x + M (A^t)^{-1} y) \cdot ((A^t)^{-1}y) = x \cdot y,
\end{equation}
and so the general problem amounts to counting solutions of $x \cdot y =
\ell$ with
\begin{equation}
(x,y) \in B_r' := \{(x,y) \in (\R^c)^2 : |A x + M (A^t)^{-1} y|^2+|(A^t)^{-1}y|^2 \le r^2\}.
\end{equation}
The only difference in this calculation is in the value
$\sigma_\infty(B_r')$, but $\sigma_\infty$ is an $\textup{O}(c,c)$-invariant measure.
Because $B_r'$ is the image of $B_r$ under an element of $\textup{O}(c,c)$, we obtain
the same constant for any Narain lattice, which completes the informal
derivation of Siegel's theorem.

\subsection{Modular invariance} \label{ss:modinv}

It is instructive to rephrase this derivation in terms of the partition
function. Doing so amounts to a weighted version of the circle method, and it
highlights the role of modular invariance in dealing with the major arcs. In
this calculation we take $\btau = \tau^*$.

Define the reduced partition function by
\begin{equation}
\widehat{Z}(\tau, \tau^*) = (\Im\tau)^{c/2} |\eta(\tau)|^{2c} Z(\tau,\tau^*) ,
\end{equation}
where the $|\eta(\tau)|^{2c}$ factor removes the denominator from the
characters while the $(\Im\tau)^{c/2}$ factor restores modular invariance.
Its leading behavior as  $\Im \tau \to \infty$, the vacuum contribution, is
\begin{equation} \label{eq:Zhatinf}
\widehat{Z}(\tau, \tau^*) \sim (\Im \tau)^{c/2} .
\end{equation}
Our goal is to show that this vacuum term is responsible for the asymptotics
\eqref{Narain-asymp} using modular invariance. We will again break up an
integral into contributions from major arcs, and dealing with them will
require asymptotics for $\widehat{Z}(\tau,\tau^*)$ near rational numbers
$a/b$, or equivalently cusps of $\textup{SL}(2,\Z)$. Specifically, suppose
$\gcd(a,b)=1$, and we wish to approximate $\widehat{Z}(\tau,\tau^*)$ for
$\tau$ near $a/b$, i.e., $\tau = a/b+x+yi$ with $x$ and $y$ small. By
choosing integers $f$ and $g$ with $af+bg=-1$, we obtain a matrix
\begin{equation}
\begin{pmatrix} f & g \\ b & -a \end{pmatrix}
\end{equation}
in $\textup{SL}(2,\Z)$, which maps $a/b$ to $i \infty$. It maps nearby points
$a/b+x+yi$ to
\begin{equation}
\frac{f}{b} - \frac{x}{b^2(x^2+y^2)} + \frac{y}{b^2(x^2+y^2)} i,
\end{equation}
whose imaginary part tends to infinity as we approach $a/b$. When $x$ and $y$
are both small, we conclude from modular invariance and \eqref{eq:Zhatinf}
that
\begin{equation} \label{eq:Zhatcusps}
\widehat{Z}(a/b + x + yi, a/b + x -yi) \sim \left(\frac{y}{b^2(x^2+y^2)}\right)^{c/2}.
\end{equation}

We will use this approximation in a manner similar to Cardy's calculation of
the total density of states in a CFT \cite{Cardy:1986ie}, but refined to
project onto an individual spin.\footnote{Related ideas have been discussed
recently in the Virasoro context \cite{Alday:2019vdr, Pal:2019zzr}.
See also \cite{Manschot:2007ha,Cheng:2011ay,Cheng:2012qc,Ferrari:2017msn} for
supersymmetric versions (where the partition function is holomorphic) and
\cite{Mukhametzhanov:2018zja} for related applications to conformal
correlators.} By an inverse Fourier transform, the density of primaries
$\rho_{\ell}(\Delta)$ obeys
\begin{equation}
Z_\ell(y) := y^{c/2} \int d\Delta\, e^{-2\pi y \Delta} \rho_{\ell}(\Delta) = \int_0^{1} dx\, e^{-2\pi i \ell x} \widehat{Z}(x+yi,x-yi).
\end{equation}
The asymptotic density of primaries is encoded in the behavior of $Z_\ell(y)$
as $y \to 0$. The dominant contribution to the integral in this regime comes
from the major arcs and can be described as follows (see, for example,
\cite{Heath-Brown}). Let $B$ be a bound depending on $y$, with $B \to
\infty$, $B^3y \to 0$, and $B^4y \to \infty$ as we take $y \to 0$. Using the
major arcs, we approximate $Z_\ell(y)$ by
\begin{equation}\label{hblemma}
Z_\ell(y) \sim \sum_{1 \le b \leq B}  \sum_{ \substack{1 \le a \le b\\ \gcd(a,b)=1}} \int_{-1/(bB^2)}^{1/(bB^2)} dx \,
e^{-2\pi i \ell(a/b + x)}\widehat{Z}\left( \frac{a}{b} + x + yi,  \frac{a}{b} + x - yi\right) .
\end{equation}

Within the range of integration in \eqref{hblemma}, our assumptions on $B$
imply that we can use \eqref{eq:Zhatcusps} to estimate $\widehat{Z}$. The
phase $e^{-2\pi i \ell x}$ under the integrand is approximately constant,
and so
\begin{equation}
\begin{split}
Z_\ell(y) &\sim \sum_{1 \le b \leq B}  \sum_{ \substack{1 \le a \le b\\ \gcd(a,b)=1}}  e^{-2\pi i \ell a/b} b^{-c} y^{c/2} \int_{-1/(bB^2)}^{1/(bB^2)} \frac{dx}{(x^2+y^2)^{c/2}}\\
&=  \sum_{1 \le b \leq B}  \sum_{ \substack{1 \le a \le b\\ \gcd(a,b)=1}}  e^{-2\pi i \ell a/b} b^{-c} y^{1-c/2} \int_{-1/(bB^2y)}^{1/(bB^2y)} \frac{du}{(1+u^2)^{c/2}}.
\end{split}
\end{equation}
Now our assumption that $B^3 y \to 0$ implies that the integral converges to an
integral over the entire line, which we can evaluate using the beta function
as
\begin{equation}
\int_{-\infty}^{\infty} \frac{du}{(1+u^2)^{c/2}} = \frac{\pi^{1/2} \Gamma( \frac{c-1}{2} )}{\Gamma(c/2)}.
\end{equation}
Thus, we have found that
\begin{equation}
Z_\ell(y) \sim y^{1-c/2} \frac{\pi^{1/2} \Gamma( \frac{c-1}{2} )}{\Gamma(c/2)} \sum_{b=1}^\infty  \sum_{ \substack{1 \le a \le b\\ \gcd(a,b)=1}}  e^{-2\pi i \ell a/b} b^{-c}.
\end{equation}
Ramanujan \cite[\S9.6]{MR0106147} showed that
\begin{equation}
\sum_{b=1}^\infty  \sum_{ \substack{1 \le a \le b\\ \gcd(a,b)=1}}  e^{-2\pi i \ell a/b} b^{-c} = \frac{\sigma_{1-c}(\ell)}{\zeta(c)} ,
\end{equation}
from which we conclude that
\begin{equation}
Z_{\ell}(y)
\sim \frac{\pi^{1/2} \sigma_{1-c}(\ell) \Gamma( \frac{c-1}{2} ) }{ \Gamma(c/2)\zeta(c)} y^{1-c/2}
\end{equation}
as $y \to 0$.

In terms of the density of states,
\begin{equation}
\int d\Delta \, \rho_\ell(\Delta) e^{-2\pi \Delta y} \sim
 \frac{\pi^{1/2} \sigma_{1-c}(\ell) \Gamma( \frac{c-1}{2} ) }{ \Gamma(c/2)\zeta(c)} y^{1-c}
\end{equation}
as $y \to 0$, and the inverse Laplace transform of the right side is
\begin{equation}
\frac{2\pi^c \sigma_{1-c}(\ell) \Delta^{c-2} }{\Gamma(c/2)^2 \zeta(c)} .
\end{equation}
From Karamata's Tauberian theorem \cite[Theorem~4.3 of Chapter~V]{Widder} we
conclude that this quantity is the density of states in an averaged sense as
$\Delta \to \infty$. That is,
\begin{equation}
\int_{|\ell|}^\Delta d\tilde{\Delta}\, \rho_{\ell}(\tilde{\Delta})  \sim \frac{2\pi^c \sigma_{1-c}(\ell) \Delta^{c-1} }{(c-1)\Gamma(c/2)^2 \zeta(c)}
\end{equation}
as $\Delta \to \infty$, which gives precisely the constant in Siegel's theorem.

\subsection{Spectral gap} \label{ss:spectralgap}

Theorem~\ref{theorem:siegel} proves the existence of Narain CFTs with
spectral gap
\begin{equation}
\Delta_1 = (1+o(1)) \frac{c}{2\pi e}
\end{equation}
as $c \to \infty$. The reasoning is simple: the expected number of non-vacuum
primary states with $\Delta \le \alpha c$ is
\begin{equation}
\frac{2\pi^c\sigma_{1-c}(\ell)}{\Gamma(c/2)^2\zeta(c)} \sum_{|\ell| \le  \alpha c} \int_{|\ell|}^{ \alpha c} d\Delta \, (\Delta^2 - \ell^2)^{c/2-1},
\end{equation}
which is at most a constant times
\begin{equation}
\frac{\pi^c}{\Gamma(c/2)^2} \sum_{|\ell| \le  \alpha c} \frac{( \alpha c)^{c-1}}{c-1},
\end{equation}
and thus at most a constant times
\begin{equation}
\frac{(\pi  \alpha c)^c}{\Gamma(c/2)^2}.
\end{equation}
Stirling's formula shows that this bound is
\begin{equation}
(2 \pi e \alpha + o(1))^c
\end{equation}
as $c \to \infty$. If $\alpha < 1/(2\pi e)$, then the expected number of
states tends to $0$ as $c \to \infty$. Because the number of primaries with
$\Delta \le \alpha c$ is always an integer, it must vanish for some Narain
CFTs, in fact almost all of them. Letting $\alpha \to 1/(2\pi e)$ as $c \to
\infty$ shows that we can obtain $\Delta_1 = (1+o(1)) c/(2\pi e)$, as
desired.

This sort of averaging argument cannot prove any better bound for the
spectral gap: if $\alpha > 1/(2\pi e)$, then the expected number of states
grows exponentially, and we cannot rule out the possibility that every Narain
CFT has at least one non-vacuum primary in this range. In sphere packing
terms, $\alpha = 1/(2\pi e)$ corresponds to the Minkowski-Hlawka lower bound
for the sphere packing density (namely, a lower bound of $2^{-d}$ in $\R^d$),
which is the best lower bound known up to subexponential factors. Because all
Narain lattices yield sphere packings, any improvement on $1/(2\pi e)$ would
yield exponentially denser sphere packings and thus solve a longstanding open
problem in discrete geometry.

\section{Holographic duality}\label{s:holography}

In this section we set $\btau = \tau^*$, so that the CFT partition function
is equal to the Euclidean path integral on a torus with modulus $\tau$.

\subsection{Warm-up: the $\textup{U}(1)^c$ Cardy formula}

The conclusion that averaged Narain lattices have $\Delta_1 \sim c/(2\pi e)$
suggests a holographic interpretation. First we will  aim to provide some
intuition for this connection, while postponing the more careful analysis to
the next subsection.

Before turning to the $\textup{U}(1)^c$ case, consider a CFT with only Virasoro
symmetry. We specialize to zero angular potential, i.e.,  $\tau = - \btau = i
\beta$ with $\beta$ the inverse temperature. At high temperature, or
equivalently $\beta \to 0$, the partition function can be approximated by
doing an $S$ transformation and keeping only the vacuum state in the dual
channel, which yields
\begin{equation}
Z(\beta) = Z(-1/\beta) \approx e^{\pi c / (6\beta)} .
\end{equation}
Re-expressed in the original channel, this approximation corresponds to the
Cardy \cite{Cardy:1986ie} density of states
\begin{equation}\label{rhocardy}
\rho_\textup{Cardy} (\Delta)  \approx \exp\left( 2\pi \sqrt{\frac{c}{3}\left(\Delta - \frac{c}{12}\right)}\right) ,
\end{equation}
where we have kept only the exponential dependence. In a general CFT, this
formula controls the average asymptotic density of states as $\Delta \to
\infty$, and applies only for $\Delta \gg c$. However, in a holographic CFT
dual to pure gravity in three dimensions, the Cardy regime is extended. In
these theories, \eqref{rhocardy} applies for $\Delta \gtrsim \frac{c}{12}$,
and this formula should be viewed as a large-$c$ limit rather than a
large-$\Delta$ limit \cite{Hartman:2014oaa}.\footnote{See
\cite{Benjamin:2015hsa} for a related discussion of elliptic genera in
supersymmetric theories.} In the gravitational theory, \eqref{rhocardy} is
interpreted as the density of states of the BTZ black hole
\cite{Strominger:1997eq}.

In a theory of pure 3d gravity, we may expect the first nontrivial primary
state to be a black hole microstate, so that $\Delta_1 \sim \frac{c}{12}$
(although it could be lower; see \cite{Benjamin:2019stq,Alday:2019vdr,Benjamin:2020mfz}). In fact, quite
generally the physics of pure gravity in three dimensions is captured by the
contribution of the vacuum conformal block in different channels
\cite{Yin:2007gv,Hartman:2013mia,Hartman:2014oaa}. The conclusion is that in
the CFT dual, we can estimate the spectral gap to be the value of $\Delta$ at
which the Cardy density of states becomes large.

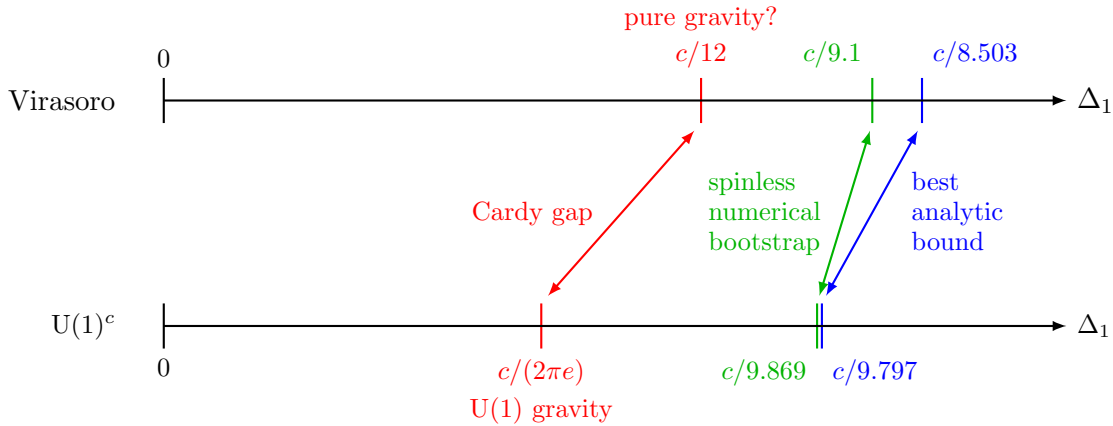
\begin{figure}
\begin{center}
\begin{tikzpicture}[thick]
\draw (12,3) node[right] {$\Delta_1$};
\draw (-0.5,3) node[left] {Virasoro};
\draw (0,2.7)--(0,3.3); \draw (0,3.3) node[above] {\small $0$};
\draw[red] (7.1428571428,2.7)--(7.1428571428,3.3); \draw[red] (7.1428571428,3.3) node[above] {\small $c/12$};
\draw[red] (7.1428571428,3.8) node[above] {\small pure gravity?};
\draw[green!70!black] (9.419152276295,2.7)--(9.419152276295,3.3); \draw[green!70!black] (9.419152276295,3.3) node[above left] {\small $c/9.1$};
\draw[blue] (10.080066,2.7)--(10.080066,3.3); \draw[blue] (10.080066,3.3) node[above right] {\small $c/8.503$};
\draw[red,<->,>=latex] (5.1070695,0.40000000)--(7.0543446,2.6000000);
\draw[red] (4.9,1.5) node {\small Cardy gap};
\draw[green!70!black,<->,>=latex] (8.7157852,0.40000000)--(9.3885711,2.6000000);
\draw[green!70!black] (7.1,1.9) node[right] {\small spinless};
\draw[green!70!black] (7.1,1.5) node[right] {\small numerical\vphantom{p}}; 
\draw[green!70!black] (7.1,1.1) node[right] {\small bootstrap};
\draw[blue,<->,>=latex] (8.8047106,0.40000000)--(10.024616,2.6000000);
\draw[blue] (9.8,1.9) node[right] {\small best\vphantom{y}}; 
\draw[blue] (9.8,1.5) node[right] {\small analytic\vphantom{b}}; 
\draw[blue] (9.8,1.1) node[right] {\small bound\vphantom{y}}; 
\draw (12,0) node[right] {\small $\Delta_1$};
\draw (-0.5,0) node[left] {\small $\textup{U}(1)^c$};
\draw (0,-0.3)--(0,0.3); \draw (0,-0.3) node[below] {\small $0$};
\draw[red] (5.018556987,-0.3)--(5.018556987,0.3); \draw[red] (5.018556987,-0.3) node[below] {\small $c/(2 \pi e)$};
\draw[red] (5.018556987,-0.8) node[below] {\small $\textup{U}(1)$ gravity};
\draw[green!70!black] (8.685204,-0.3)--(8.685204,0.3); \draw[green!70!black] (8.685204,-0.3) node[below left] {\small $c/9.869$};
\draw[blue] (8.74926038,-0.3)--(8.74926038,0.3); \draw[blue] (8.74926038,-0.3) node[below right] {\small $c/9.797$};
\draw[->,>=latex] (0,3)--(12,3);
\draw[->,>=latex] (0,0)--(12,0);
\end{tikzpicture}
\end{center}
\caption{Spectral gap for the Virasoro algebra and $\textup{U}(1)^c$ algebra at large central charge $c$. The green and blue marks show upper bounds on $\Delta_1$ from linear programming, i.e., the modular bootstrap. The numerical upper bounds were estimated for Virasoro in \cite{Afkhami-Jeddi:2019zci}  and for $\textup{U}(1)^c$ in \cite{paper1}. The analytic result for $\textup{U}(1)^c$ is the Kabatyanskii-Levenshtein bound \cite{KL,MR3229046}, and the analytic bound for Virasoro was derived in \cite{HMR}.}
\label{figure:best-bounds}
\end{figure}

Now let us repeat this analysis for a theory with the chiral algebra
$\textup{U}(1)^c$. The situation is summarized in Figure~\ref{figure:best-bounds},
along with numerical and analytic upper bounds on the spectral gap. The
analogue of the Cardy formula \cite{paper1} for $\textup{U}(1)^c$ is
\begin{equation}
\rho_\textup{Cardy}(\Delta) \sim \frac{(2\pi)^c \Delta^{c-1}}{\Gamma(c)},
\end{equation}
which has support down to $\Delta = 0$. However, that does not mean the
spectral gap is zero, because for small $\Delta$ there is on average less
than one state. To estimate the spectral gap we set
$\rho_\textup{Cardy}(\Delta_1) \approx 1$ and take the $c \to \infty$ limit.
The result is
\begin{equation}
\Delta_1 \sim \frac{c}{2\pi e} .
\end{equation}
This calculation agrees with the spectral gap of an average Narain lattice
from Section~\ref{ss:spectralgap}. In other words, an average Narain lattice
saturates the Cardy estimate for $\Delta_1$.

This coincidence suggests looking for a holographic dual. This argument is
certainly not conclusive, though. In particular, we would not expect the
holographic dual to have black holes that dominate the canonical ensemble at
finite temperature. In other words, other modular images of the vacuum under
$\textup{SL}(2, \mathbb{Z})$ can be equally important. To check whether the
holographic interpretation survives a more careful analysis we will now
examine these other contributions.

\subsection{Bulk partition function} \label{ss:bpf}
Let us calculate the partition function of the three-dimensional theory of
$\textup{U}(1)$ gravity described in the introduction. As we have stressed, we do not
have a full non-perturbative definition of this theory. In the introduction
we have only specified its perturbative excitations on a torus, and now we
will give a prescription to calculate the genus-one partition function by a
sum over topologies.

The first step is to calculate the perturbative contribution in thermal
AdS$_3$. The theory is topological, so the metric makes no difference, but we
will nevertheless refer to these manifolds in the language of AdS/CFT to make
the analogy clear. Thermal AdS$_3$ is a hyperbolic 3-manifold with a torus
conformal boundary and the topology of a solid torus. To describe it, let $z$
be a coordinate on the boundary torus, with the identifications
\begin{equation}
z \sim z + \tau \sim z + 1 .
\end{equation}
Thermal AdS$_3$ is by definition the hyperbolic manifold filling in this
torus with the cycle $z \sim z + 1$ contractible in the bulk.

The 1-loop partition function of $\textup{U}(1)$ gravity in thermal AdS$_3$ is by
design equal to the $\textup{U}(1)^c \times \textup{U}(1)^c$ vacuum character,
\begin{equation}\label{ztads}
Z_\textup{tAdS}(\tau, \btau) =\frac{1}{\eta(\tau)^c \eta(-\btau)^c}  = \chi_0(\tau) \bar{\chi}_0(\btau) .
\end{equation}
This formula is derived in \cite{Porrati:2019knx} from the 1-loop determinant
of the Chern-Simons fields (including the contributions from gauge fixing).
It is also easy to understand from a Hamiltonian point of view, because the
bulk theory has asymptotic symmetries corresponding to the $\textup{U}(1)^c \times
\textup{U}(1)^c$ affine algebra. The theory is quadratic, so the result is exact in
perturbation theory.\footnote{ The perturbative calculation is insensitive to
the global structure of the gauge group, so we can take it to be non-compact.
In other words, we are not performing an additional sum over nontrivial gauge
configurations.}

The full partition function is a sum over topologies with the boundary
condition $\tau$ at infinity:
\begin{equation}
Z(\tau, \btau) = \sum_\textup{topologies}Z_{\mathcal{M}}(\tau,\btau) .
\end{equation}
We will sum over the solid tori obtained by filling in different cycles of
the boundary torus, as in \cite{Maloney:2007ud}.  In gravity language, we sum
over the Euclidean BTZ black holes. It is not obvious why this is the right
thing to do, and it is a provisional choice motivated by the analogy to 3d
gravity.

The different ways of filling in the boundary torus are related by the action
of $\textup{SL}(2,\mathbb{Z})$, so roughly speaking we must sum \eqref{ztads} over
$\textup{SL}(2,\mathbb{Z})$ images. However, $Z_\textup{tAdS}$ is invariant under
$\tau \to \tau  + 1$, so these contributions are not distinct. The distinct
contributions are labeled by elements of $\textup{SL}(2, \mathbb{Z})/\Gamma_\infty$,
where $\Gamma_\infty$ is generated by $T$. Thus
\begin{equation}
Z(\tau, \btau) = \sum_{\gamma \in \textup{SL}(2,\mathbb{Z})/\Gamma_\infty} \frac{1}{|\eta(\gamma\tau)|^{2c}} = (\Im \tau)^{-c/2} |\eta(\tau)|^{-2c} \sum_{\gamma \in \textup{SL}(2,\mathbb{Z})/\Gamma_\infty} (\Im \gamma \tau)^{c/2}  ,
\end{equation}
where in the second equation we used the fact that the combination $(\Im
\tau)^{1/2} |\eta(\tau)|^2$ is modular invariant.

This sum is proportional to a non-holomorphic Eisenstein series. That is,
\begin{equation}
Z(\tau, \btau) = (\Im \tau)^{-c/2} |\eta(\tau)|^{-2c} E\mathopen{}\left(\tau, \frac{c}{2}\right)\mathclose{},
\end{equation}
where the Eisenstein series is defined by
\begin{equation}
E(\tau, s) = \sum_{\gamma \in \textup{SL}(2,\mathbb{Z})/\Gamma_\infty} (\Im \gamma \tau)^s  .
\end{equation}
For $c>2$, the sum converges. Siegel proved that in this case $Z(\tau,\btau)$
agrees with the CFT partition function averaged over moduli
\cite{MR67930,MR67931,MR0271028}.\footnote{ See Theorem~12 in Chapter~4 of
\cite{MR0271028}.} To reproduce his result, we will extract the spectrum
$\rho_\ell(\Delta)$ from $Z(\tau,\btau)$ by comparing to the general form
\eqref{zld}. We first do a Fourier transform to organize by spin $\ell$, then
an inverse Laplace transform to find $\rho_\ell(\Delta)$. The Fourier
expansion of the non-holomorphic Eisenstein series is (see, for example,
\cite[Section~5.2]{MR3675870})
\begin{equation}\label{efourier} \begin{split}
E(\tau, s) &= y^s + \frac{\pi^{1/2} \Gamma(s - \frac{1}{2})\zeta(2s-1)}{\Gamma(s)\zeta(2s)}y^{1-s}\\
&\quad   \phantom{}+\frac{4\pi^s}{\Gamma(s) \zeta(2s)}\sum_{\ell=1}^{\infty} \ell^{s-1/2}\sigma_{1-2s}(\ell) y^{1/2} K_{s-1/2}(2\pi \ell y)\cos(2\pi \ell x)
\end{split}
\end{equation}
with $\tau = x + yi$, $K_\nu$ the modified Bessel function, and
\begin{equation}
\sigma_t(\ell) = \sum_{n|\ell}n^t
\end{equation}
the divisor function. Taking the inverse Laplace transform of the first two
terms in \eqref{efourier} and comparing to \eqref{zld} gives the scalar
density of states
\begin{equation}\label{scalarrho}
\rho_0(\Delta) = \delta(\Delta)  +  \frac{2\pi^{c}\zeta(c-1)}{\Gamma(\frac{c}{2})^2\zeta(c)}\Delta^{c-2} .
\end{equation}
The delta function at zero is the vacuum state.\footnote{We use the convention
$\int_0^\infty d\Delta \,\delta(\Delta) = 1$.} After an inverse Laplace
transform, the spinning terms in \eqref{efourier} lead to
\begin{equation}\label{spinningrho}
\rho_{\ell}(\Delta) = \frac{2\pi^c \sigma_{1-c}(\ell)}{\Gamma(c/2)^2\zeta(c)}
(\Delta^2 - \ell^2)^{c/2-1} .
\end{equation}
The results \eqref{scalarrho}--\eqref{spinningrho} agree exactly with the
density of states of an averaged Narain lattice from
Theorem~\ref{theorem:siegel}.

\subsection{Origin of the agreement}

We have reproduced Siegel's result relating the Eisenstein series to an
integrated partition function by explicitly calculating both sides and
comparing term by term. A more conceptual explanation is as follows. In the
derivation of the Siegel measure on random Narain lattices
in Section~\ref{ss:modinv}, we argued that
there is unique modular invariant partition function that is homogenous on
each hyperboloid $|u|^2-|v|^2=2\ell$ with $(u,v) \in (\R^c)^2$. That is, any
modular-invariant spectrum with $\rho_\ell(\Delta) \propto
(\Delta^2-\ell^2)^{c/2-1}$ and a unique vacuum state will necessarily agree
with a random Narain lattice. The circle method calculation to determine the prefactors for
each spin orbit depends only on the asymptotics of the partition function,
and these asymptotics are fixed by modular invariance.

The Eisenstein series is modular invariant by construction for $c>2$, so we only
need to check that $\rho_\ell(\Delta)$ has the correct dependence on
$\Delta$. To this end, we will use the fact that the Eisenstein series is a
Maass form, i.e., an automorphic eigenfunction of the hyperbolic Laplacian on
the upper half-plane. Let
\begin{equation}
\Delta_{H} = -y^2 \left(\frac{\partial^2}{\partial x^2} + \frac{\partial^2}{\partial y^2} \right),
\end{equation}
where $\tau = x+yi$. This operator is invariant under $\textup{SL}(2,\mathbb{Z})$ and
satisfies
\begin{equation}
\Delta_H (\Im \tau)^s = s(1-s)(\Im \tau)^s .
\end{equation}
It follows that the Eisenstein series is also an eigenfunction, with
\begin{equation}\label{eigenE}
\Delta_H E(\tau, s) = s(1-s) E(\tau,s) .
\end{equation}
Now we examine the consequences for the partition function
\begin{equation}
Z = (\Im \tau)^{-c/2} |\eta(\tau)|^{-2c} E(\tau, c/2)  .
\end{equation}
The eigenvalue equation \eqref{eigenE} implies
\begin{equation}\label{eigenZ}
\Delta_H \left( y^{c/2} |\eta(\tau)|^{2c} Z \right) = \frac{c}{2}\left(1- \frac{c}{2}\right) y^{c/2} |\eta(\tau)|^{2c} Z ,
\end{equation}
and the expansion of $Z$ in $\textup{U}(1)^c \times \textup{U}(1)^c$ characters yields
\begin{equation}
|\eta(\tau)|^{2c} Z = \sum_{\ell=-\infty}^{\infty} \int_{|\ell|}^{\infty}  d\Delta\,  e^{-2\pi y \Delta + 2\pi i x \ell}\rho_{\ell}(\Delta) .
\end{equation}
The key identity is
\begin{equation}\label{keyid}
\left(\Delta_H -\frac{c}{2}\left(1-\frac{c}{2}\right)
+ c \Delta \p_\Delta + (\Delta^2-\ell^2)\p_\Delta^2 \right)\left( e^{-2\pi y \Delta + 2 \pi i x \ell}y^{c/2} \right) = 0,
\end{equation}
where $\p_\Delta$ denotes differentiation with respect to $\Delta$.
The operator $c \Delta \p_\Delta + (\Delta^2-\ell^2)\p_\Delta^2$ is
proportional to the Laplacian $\del^2_{\ML}$ on the $(2c-1)$-dimensional
hyperboloid $\ML = \{(u,v) \in (\R^c)^2: |u|^2-|v|^2=2\ell \}$ with the
metric in \eqref{hypermetric}, acting on a function of $|u|$ or equivalently
$\Delta = \frac{1}{2}(|u|^2 + |v|^2) = |u|^2-\ell$. Specifically, the
Laplacian acts on such functions by
\begin{equation}
\del^2_{\ML} \propto - \frac{1}{\sqrt{|G|}} \frac{\p}{\p \alpha} \left( \sqrt{|G|} \frac{\p}{\p \alpha}\right) = -4 (c\Delta \p_\Delta   + (\Delta^2-\ell^2)\p_\Delta^2) .
\end{equation}
Here $\sqrt{|G|} \propto (|u||v|)^{c-1}$ is the volume factor on
$\ML$ obtained below equation \eqref{hypermetric}.\footnote{Equation
\eqref{keyid} is a consequence of Howe duality \cite{MR538856,MR986027}. See
\cite[Section~III.2.3]{MR1151617} for a pedagogical discussion.} Using
\eqref{keyid} in the eigenvalue equation \eqref{eigenZ} and projecting onto
an individual spin $\ell$ gives
\begin{equation}
0 = \int_{|\ell|}^\infty d\Delta \,\rho_\ell(\Delta) \left( c\Delta \p_\Delta   + (\Delta^2-\ell^2) \p_{\Delta}^2 \right) e^{-2\pi y \Delta}  .
\end{equation}
Integrating by parts now yields
\begin{equation} \label{bp2}
\rho_\ell(|\ell|)(2-c)|\ell| e^{-2\pi y |\ell|} + \int_{|\ell|}^{\infty} d\Delta \, (D \rho_\ell(\Delta))e^{-2\pi y\Delta}
 = 0,
\end{equation}
where
\begin{equation}
D = 2-c+(4-c)\Delta \p_\Delta +  (\Delta^2-\ell^2)\p_\Delta^2.
\end{equation}
By acting on \eqref{bp2} with $\p_y + 2\pi |\ell|$, we can remove the first term and obtain
\begin{equation}
 \int_{|\ell|}^{\infty} d\Delta\, (\Delta-|\ell|)(D \rho_\ell(\Delta))e^{-2\pi y\Delta}
 = 0,
\end{equation}
from which we conclude that $D \rho_\ell(\Delta) = 0$. This equation expresses the
requirement that $\rho_\ell(\Delta)$ is proportional to a covariantly constant scalar density
on the hyperboloid $\ML$. The solution to $D \rho_\ell(\Delta) = 0$ that vanishes
at $\Delta = |\ell|$ is  $\rho_\ell(\Delta)\propto(\Delta^2 - \ell^2)^{c/2-1}$,
and the other solution does not satisfy \eqref{bp2} when $\ell\ne0$ because of the boundary term.
(When $\ell=0$, the other solution is $\rho_0(\Delta) \propto 1/\Delta$, which is not
integrable near $\Delta=0$.)
Thus $\rho_\ell(\Delta)$ is proportional to the volume factor $(\Delta^2-\ell^2)^{c/2-1}$,
which is exactly what we needed to conclude that the full spectrum agrees
with the average Narain CFT.

To summarize, the fact that the Eisenstein series is an eigenfunction of the
Laplacian on the upper half plane implies that the spectrum for each $\ell$
is proportional to the volume element on the hyperboloid $\ML$, and then
modular invariance fixes the full spectrum.

\subsection{Comments}
The density of states we have obtained is manifestly positive, unlike the
analogous result in pure gravity \cite{Maloney:2007ud}. We interpret the
continuous spectrum as a consequence of ensemble averaging.  The spectrum
extends all the way down to the unitarity bound $\Delta=|\ell|$, with the
low-energy contributions on the CFT side coming from the decompactification
limit in the Narain moduli space. Note, however, that at large central
charge, it is very rare to find primary states with $\Delta \ll c$ other than
the vacuum.

Non-compact, non-averaged CFTs also have a continuous spectrum. However it
seems impossible to interpret \eqref{scalarrho} in this way, because of the
delta function corresponding to the vacuum state. The vacuum is not present
as a normalizable state in a non-compact CFT, but is present in an averaged
compact CFT.

In \cite{Witten:1988hf} Witten established an exact equivalence between
Chern-Simons gauge theory and rational CFT. With an abelian gauge group,
Witten's correspondence gives a three-dimensional realization of a Narain CFT
at rational points in moduli space (see
\cite{Elitzur:1989nr,Polychronakos:1989cd,Manoliu:1996fx}). The dictionary
for this duality differs from that of AdS/CFT, so it is not a holographic
duality in the usual sense. A direct connection to AdS/CFT was made in
\cite{Gukov:2004id}, where a compact abelian Chern-Simons theory in AdS$_3$
was related to a rational Narain CFT following the usual holographic
dictionary. It is not clear exactly how either of these results is related to
the duality conjectured in the present paper.  Note that before doing the sum over
topologies, $\textup{U}(1)$ gravity is not dual to an individual member of the ensemble of
Narain CFTs, while the construction of \cite{Gukov:2004id} does provide such a duality.
Perhaps this construction can be used to
define alpha states of $\textup{U}(1)$ gravity in the sense of
\cite{Coleman:1988cy,Giddings:1988cx,Marolf:2020xie}.

\bigskip 

\section*{Acknowledgements}

We thank Nathan Benjamin, Scott Collier, Kristan Jensen, David de Laat, Greg Moore, Hirosi Ooguri, Natalie Paquette, Leonardo Rastelli, Peter Sarnak, Eva Silverstein, David Simmons-Duffin, and Douglas Stanford for useful discussions. The work of TH and AT is supported by the Simons
Foundation (Simons Collaboration on the Nonperturbative Bootstrap). NA is
supported by the Leo Kadanoff Fellowship. This work was completed in part
with resources provided by the University of Chicago Research Computing Center.

\appendix

\section{Details of numerical bootstrap}\label{appendix:sdpb}

At a given $\Delta_\textup{gap}$, the infinite set of positivity constraints
\eqref{opos1}--\eqref{opos2} can be recast as a semidefinite program with an
infinite sequence of constraints labeled by spin \cite{Poland:2011ey}. In
practice, a functional satisfying all the constraints can be obtained even if
we truncate to a finite set of spins, such as
\begin{equation}
h - \bh = 0, 1, \dots, L_1, L_2,
\end{equation}
for some large $L_1$, $L_2$. That is, once $\Delta_\textup{gap}$ is tuned to
its optimum, the resulting functional is found to automatically obey the
higher spin constraints that were not included in this list. The
computational problem now takes the standard form of a semidefinite program
that can be optimized by a numerical solver. We use SDPB v1.0
\cite{Simmons-Duffin:2015qma}, which is designed to take advantage of the
special structure in a semidefinite program organized by spin.

We fix $\Delta_\textup{gap}$ and run SDPB to determine whether the
constraints can be satisfied; we then adjust $\Delta_\textup{gap}$ by
bisection to find the optimal bound at truncation order $K$. We have
generated bounds at $K=17$, $19$, $21$, $23$, and $25$ for $1\leq c \leq 15$.
To generate functionals that obey all of the positivity conditions requires
many bisection steps. To save computational time, we ran only $K=19$ at a
high level of rigor: in this case we set $L_1 = 50$, $L_2=100$, and ran a
large number of bisections. The resulting functionals obey all of the
constraints. For other values of $K$, we set $L_1 = 20$, $L_2 = 30$, and ran
fewer bisections. The resulting functionals do not obey all of the
constraints at high spin, but from experience we expect them to be accurate
nonetheless. The numerical functionals at $K=19$ can be downloaded from
\url{https://hdl.handle.net/1721.1/125646}.

The spinning bootstrap is much more computationally intensive than the
spinless bootstrap. This is partly because we are now optimizing over a
two-dimensional space of functionals, and partly because at present there is
no algorithm based on forced roots to bypass linear programming.  We
therefore find a good estimate of the bound only for $c \lesssim 10$, as
compared to $c \lesssim 1000$ for the spinless bounds in previous work
\cite{Afkhami-Jeddi:2019zci,paper1}.

The SDPB settings we used are listed in Table~\ref{table:sdpbsettings}. SDPB
also requires a normalization condition and a set of sampling points. Our
normalization condition sets the coefficient of $f_{1,0}$ to $1$, and the
sampling points are the defaults in the Mathematica package provided with
SDPB.

\begin{table}
\begin{center}
\texttt{
\begin{tabular}{l|l}
findPrimalFeasible & false \\ \hline
findDualFeasible & true \\ \hline
detectPrimalFeasibleJump & false \\ \hline
detectDualFeasibleJump & false \\ \hline
precision & 500 \\ \hline
dualityGapThreshold & 1e-15\\ \hline
primalErrorThreshold & 1e-100 \\ \hline
dualErrorThreshold & 1e-100 \\ \hline
initialMatrixScalePrimal & 1e20 \\ \hline
initialMatrixScaleDual & 1e20 \\ \hline
feasibleCenteringParameter & 0.1 \\ \hline
infeasibleCenteringParameter & 0.3 \\ \hline
stepLengthReduction & 0.7 \\ \hline
choleskyStabilizeThreshold  & 1e-40 \\ \hline
maxComplementarity & 1e80
\end{tabular}
}
\caption{SDPB runtime parameters. \label{table:sdpbsettings}}
\end{center}
\end{table}

\section{Details of optimal Narain lattices}
\label{appendix:narain}

Let $\langle \cdot,\cdot \rangle$ denote the Euclidean inner product on
$\R^{2c}$, and let $[\cdot,\cdot]$ denote the usual bilinear form of
signature $(c,c)$; i.e., $\langle x,x \rangle = \sum_{i=1}^{2c} x_i^2$ and
$[x,x] = \sum_{i=1}^c x_i^2 - \sum_{i=c+1}^{2c} x_i^2$. In this notation, a
Narain lattice is an even unimodular lattice under $[\cdot,\cdot]$, which is
uniquely determined up to the action of $\textup{O}(c,c)$ but can look very different
under $\langle \cdot,\cdot \rangle$.

When one envisions a Euclidean lattice, one typically thinks about it up to
isometries, i.e., up to the action of $\textup{O}(2c)$. From this perspective, it is
not obvious which Euclidean lattices $\Lambda$ satisfy the Narain condition:
the issue is whether the $\textup{O}(2c)$-orbit of $\Lambda$ intersects the
$\textup{O}(c,c)$-orbit of the even unimodular lattice for $[\cdot,\cdot]$. For
comparison, the Leech lattice in $\R^{24}$ does not have this property,
because the spinning modular bootstrap rules it out, and it is a noteworthy
fact that the Coxeter-Todd and Barnes-Wall lattices do. We can verify it
using the following technique, which we will describe more generally in terms
of verifying the output of our computer program.

The output is a floating-point basis $b_1,\dots,b_{2c}$ for the lattice
$\Lambda$, which we would like to convert to an exact description of
$\Lambda$. There is no reason to expect the entries of these vectors to be
recognizable numbers, but the Gram matrix is generally more understandable.
Let $B$ be the matrix with $b_1,\dots,b_{2c}$ as its columns. Then the Gram
matrix of the basis with respect to the Euclidean inner product is
\begin{equation}
G := \big(\langle b_i,b_j \rangle \big)_{1 \le i,j \le 2c} = B^t B,
\end{equation}
and the Gram matrix with respect to $[ \cdot , \cdot]$ is
\begin{equation}
H := \big([ b_i,b_j ] \big)_{1 \le i,j \le 2c} = B^t D B,
\end{equation}
where $D$ is the diagonal matrix with diagonal entries $1, \dots, 1, -1,
\dots, -1$, each repeated $c$ times. By the Narain condition, the entries of
$H$ must be integers, and we can round the floating-point values to obtain
the exact matrix $H$. A priori, there is no reason to expect $G$ to be a
pleasant matrix, but for the best cases we have found with $c \le 8$ it turns
out to be proportional to an integer matrix, and the constant of
proportionality is determined by $\det(G)=1$. Thus, we can exactly identify
$G$ and $H$ in practice. Now the question is whether there is still a lattice
corresponding to these exact matrices, or whether rounding the matrices has
destroyed the lattice. The following lemma shows that the existence of a
lattice basis amounts to checking that $(G H^{-1})^2 = I$, where $I$ is the
identity matrix. Using this technique, one can verify the values of
$\Delta_1$ in Table~\ref{table:narain} rigorously.

\begin{lemma}
Let $G \in \R^{2c \times 2c}$ be a symmetric, positive definite matrix, let
$H \in \R^{2c \times 2c}$ be a symmetric matrix of signature $(c,c)$, and let
$D$ be the diagonal matrix with diagonal entries $1, \dots, 1, -1, \dots,
-1$, each repeated $c$ times. Then there exists a matrix $B \in \R^{2c \times
2c}$ such that $G = B^tB$ and $H = B^t D B$ if and only if $(G H^{-1})^2 =
I$.
\end{lemma}

\begin{proof}
By the hypotheses on $G$ and $H$, there exist invertible matrices $X,Y \in
\R^{2c \times 2c}$ such that $G = X^tX$ and $H = Y^t D Y$. Furthermore, these
equations are preserved by acting on $X$ on the left by $\textup{O}(2c)$, or on $Y$ by
$\textup{O}(c,c)$. The question is whether the $\textup{O}(2c)$-orbit of $X$ and the
$\textup{O}(c,c)$-orbit of $Y$ intersect.

If we can take $Y=X$, then
\begin{equation}
(G H^{-1})^2 = X^t X X^{-1} D (X^t)^{-1} X^t X X^{-1} D (X^t)^{-1} = I.
\end{equation}
For the converse, suppose $(GH^{-1})^2=I$. This equation is equivalent to
\begin{equation}
X^t X Y^{-1} D (Y^t)^{-1} X^t X Y^{-1} D (Y^t)^{-1} = I,
\end{equation}
and conjugating by $X^t$ shows that
\begin{equation}
X Y^{-1} D (Y^t)^{-1} X^t X Y^{-1} D (Y^t)^{-1} X^t = I.
\end{equation}
If we let $Z = XY^{-1}$, we find that $(Z D Z^t)^2 = I$. The matrix $Z D Z^t$
is symmetric, and thus by the spectral theorem there exists $U \in \textup{O}(2c)$
such that $ZDZ^t = UD'U^t$, where $D'$ is a diagonal matrix with only $1$ and
$-1$ on the diagonal. By Sylvester's law of inertia, $D$ and $D'$ must have
the same signature, and so we can take $D'=D$ without loss of generality.
Then
\begin{equation}
(U^{-1} Z) D (U^{-1} Z)^t = D,
\end{equation}
which means $U^{-1}Z \in \textup{O}(c,c)$. Because $Z=XY^{-1}$, we have obtained $U
\in \textup{O}(2c)$ and $V := U^{-1}Z \in \textup{O}(c,c)$ such that $U^{-1}X=VY$. Thus, the
$\textup{O}(2c)$-orbit of $X$ intersects the $\textup{O}(c,c)$-orbit of $Y$, as desired.
\end{proof}

In the rest of this appendix, we develop a more conceptual framework for the
Coxeter-Todd and Barnes-Wall lattices as well as more general Narain
lattices. First, we need some notation. We will write vectors in $\R^{2c}$ as
$(x,y)$ with $x,y \in \R^c$, which we interpret as column vectors for matrix
multiplication. The group $\textup{O}(c) \times \textup{O}(c)$ acts on the two components of
vectors in $\R^{2c}$, and it preserves the inner products of signatures
$(2c,0)$ and $(c,c)$. We will use $\langle \cdot,\cdot \rangle$ to denote the
Euclidean inner product on $\R^c$. Then the dual lattice $\Lambda^*$ of a
lattice $\Lambda$ in $\R^c$ is defined by
\begin{equation}
\Lambda^* = \{ x \in \R^c : \textup{$\langle x,y \rangle \in \Z$ for all $y \in \Lambda$}\}.
\end{equation}
Equivalently, if the columns of a $c \times c$ matrix $B$ form a basis for
$\Lambda$, then those of $(B^t)^{-1}$ form a basis of $\Lambda^*$.

The following proposition is a standard result about the Narain condition. It
essentially amounts to the Iwasawa decomposition for $\textup{O}(c,c)$, but we will
give a proof for the convenience of the reader.

\begin{proposition} \label{prop:narainequiv}
A lattice in $\R^{2c}$ satisfies the Narain condition if and only if it is
equivalent under the action of $\textup{O}(c) \times \textup{O}(c)$ to a lattice of the form
\begin{equation}
\left\{ \frac{(u + (M+I)v, \, u +(M-I)v)}{\sqrt{2}}: u \in \Lambda, \, v \in \Lambda^*\right\},
\end{equation}
where $\Lambda$ is a lattice in $\R^c$ and $M$ is a $c\times c$ antisymmetric matrix (i.e., $M^t=-M$).
\end{proposition}

In the CFT interpretation, $M$ is the flux of the toroidal
compactification, while the choice of metric is absorbed into $\Lambda$ and $\Lambda^*$.
It is not hard to check that such a lattice satisfies the
Narain condition (the key observation is that  $\langle Mv,v \rangle = 0$,
because $M^t=-M$), while the converse is trickier. Both directions follow
from the proof given below.

Note that the action of the $2c \times 2c$ block orthogonal matrix
\begin{equation}
T := \frac{1}{\sqrt{2}} \begin{pmatrix} I & \phantom{-}I\\ I & -I \end{pmatrix}  
\end{equation}
sends
\begin{equation}
\left\{ \frac{(u + (M+I)v, \, u +(M-I)v)}{\sqrt{2}}: u \in \Lambda, \, v \in \Lambda^*\right\}
\end{equation}
to
\begin{equation}
\{ (u + Mv, \, v) : u \in \Lambda, \,v \in \Lambda^*\}
\end{equation}
and vice versa. We will work in these coordinates, because the expressions
involve fewer symbols.

Under the action of $T$, the bilinear form $[\cdot,\cdot]$ with signature
$(c,c)$ is transformed into the form with block matrix
\begin{equation}
\begin{pmatrix} 0 & I\\ I & 0\end{pmatrix}
\end{equation}
with respect to the standard basis of $\R^{2c}$. Equivalently, the vector
$(x,y) \in \R^{2c}$ satisfies
\begin{equation}
[T(x, y),T(x,y)] = 2\langle x,y \rangle.
\end{equation}
In particular, the group $\textup{O}(c,c)$ is conjugate under $T$ to the group
\begin{equation}
G := \left\{ M \in \R^{2c \times 2c} : M^t \begin{pmatrix} 0 & I\\ I & 0\end{pmatrix} M = \begin{pmatrix} 0 & I\\ I & 0\end{pmatrix}\right\}.
\end{equation}

The lattice $\Z^{2c}$ is an even unimodular lattice under this bilinear form,
and thus all that remains is to determine the orbit of $\Z^{2c}$ under $G$.
We can do so using the following lemma.

\begin{lemma} \label{lemma:iwasawa}
Every element of $G$ can be factored as
\begin{equation}
\frac{1}{2} \begin{pmatrix} U+V & U-V \\ U-V & U+V\end{pmatrix} \cdot
\begin{pmatrix} I & M\\ 0 & I\end{pmatrix} \cdot
\begin{pmatrix} A & 0\\ 0 & (A^t)^{-1}\end{pmatrix},
\end{equation}
where $U,V \in \textup{O}(c)$ and $A$ and $M$ are $c \times c$ matrices with $\det A
\ne 0$ and $M^t = -M$.
\end{lemma}

Each of the three factors comes from a subgroup of $G$.  In particular,
\begin{equation}
\frac{1}{2} \begin{pmatrix} U+V & U-V \\ U-V & U+V\end{pmatrix}  = T \begin{pmatrix} U & 0 \\ 0 & V\end{pmatrix} T
\end{equation}
is conjugate to an element of $\textup{O}(c) \times \textup{O}(c)$ under $T$.

\begin{proof}[Proof of Proposition~\ref{prop:narainequiv}]
Given a factorization as in Lemma~\ref{lemma:iwasawa}, let $\Lambda = A
\Z^c$. Then $\Lambda^* = (A^t)^{-1} \Z^c$, and the image of $\Z^{2c}$ under
\begin{equation}
\begin{pmatrix} I & M\\ 0 & I\end{pmatrix}
\begin{pmatrix} A & 0\\ 0 & (A^t)^{-1}\end{pmatrix}
\end{equation}
is $\{ (u+Mv, \, v) : u \in \Lambda, \,v \in \Lambda^*\}$. The remaining
factor from the lemma is conjugate to an element of $\textup{O}(c) \times \textup{O}(c)$ under
$T$, which completes the proof.
\end{proof}

\begin{proof}[Proof of Lemma~\ref{lemma:iwasawa}]
Let $A,B,C,D$ be $c \times c$ matrices such that
\begin{equation}
\begin{pmatrix} A & B \\ C & D\end{pmatrix}
\end{equation}
is an element of $G$. In other words,
\begin{equation}
\begin{pmatrix} A^t & C^t \\ B^t & D^t\end{pmatrix} \begin{pmatrix} 0 & I\\ I & 0\end{pmatrix} \begin{pmatrix} A & B \\ C & D\end{pmatrix} = \begin{pmatrix} 0 & I \\ I & 0\end{pmatrix},
\end{equation}
which amounts to $A^tC+C^tA=B^tD+D^tB=0$ and $A^tD+C^tB=I$.

First, we deal with the case $C=0$. Then $A^tD=I$ and $D^tB$ is
antisymmetric. Let $M = BD^{-1}$, which is antisymmetric because $D^t B$ is
and $M = (D^{-1})^{t} (D^t  B) D^{-1}$. We conclude that our group element is
given by
\begin{equation}
\begin{pmatrix} A & B \\ 0 & D\end{pmatrix} = \begin{pmatrix} A & M(A^t)^{-1} \\ 0 & (A^t)^{-1}\end{pmatrix} = \begin{pmatrix} I & M\\ 0 & I\end{pmatrix} \begin{pmatrix} A & 0\\ 0 & (A^t)^{-1}\end{pmatrix},
\end{equation}
as desired. In this case, we have only two factors; in other words, the
missing factor is the identity matrix.

All that remains is to show we can make the lower left block of
\begin{equation}
\begin{pmatrix} A & B \\ C & D\end{pmatrix}
\end{equation}
vanish through multiplying on the left by a group element of the form
\begin{equation}
\frac{1}{2} \begin{pmatrix} U+V & U-V \\ U-V & U+V\end{pmatrix}
\end{equation}
with $U,V \in \textup{O}(c)$. The lower left block of the product is
\begin{equation}
\frac{(U-V)A+(U+V)C}{2},
\end{equation}
and so we would like to find $U,V \in \textup{O}(c)$ such that $(U-V)A+(U+V)C=0$.

Because $A^tC+C^tA = 0$, we can obtain $U$ and $V$ such that
$(U-V)A+(U+V)C=0$ by taking $U = A^t+C^t$ and $V=A^t-C^t$, but these matrices
are generally not orthogonal. In particular,
\begin{equation}
UU^t = (A^t+C^t)(A+C) = A^tA + C^tC
\end{equation}
and
\begin{equation}
VV^t = (A^t-C^t)(A-C) = A^tA + C^tC,
\end{equation}
again because $A^tC+C^tA=0$.

We can repair $U$ and $V$ as follows. No nonzero vector can be annihilated by
both $A$ and $C$, because otherwise the matrix
\begin{equation}
\begin{pmatrix} A & B \\ C & D\end{pmatrix}
\end{equation}
would not be invertible. Thus, the symmetric matrix $A^tA + C^tC$ is strictly
positive definite, and so it can be written in the form $X^t X$ for some
invertible matrix $X$. Now let $U = (X^t)^{-1}(A^t+C^t)$ and $V =
(X^t)^{-1}(A^t-C^t)$. Again $(U-V)A+(U+V)C=0$, but now
\begin{equation} \begin{split}
UU^t &= (X^t)^{-1}  (A^t+C^t)(A+C) X^{-1}\\
& = (X^t)^{-1}  (A^tA + C^tC) X^{-1}\\
& = (X^t)^{-1}  X^t X X^{-1} = I,
\end{split} \end{equation}
and similarly $VV^t=I$. Thus, $U,V \in \textup{O}(c)$, as desired.
\end{proof}

One consequence of this characterization of Narain lattices is a lower bound
for the spectral gap, which comes within a factor of $2$ of the bound
obtained in Section~\ref{ss:spectralgap}:

\begin{proposition} \label{prop:narain}
For every positive integer $c$, there exists a Narain CFT with spectral gap
\begin{equation}
\Delta_1 \ge \frac{c}{4\pi e}(1+o(1))
\end{equation}
as $c \to \infty$.
\end{proposition}

In physics terms, this bound comes from averaging over Narain CFTs with zero flux.

\begin{proof}
To prove this proposition, we will take $M=0$ in
Proposition~\ref{prop:narainequiv} and average over the choice of $\Lambda$.
Taking $M=0$ yields a lattice that is isometric to $\Lambda \times \Lambda^*$
under the Euclidean metric, and thus
\begin{equation}
\Delta_1 = \min \big(\{|x|^2/2 : x \in \Lambda\setminus\{0\}\} \cup \{|y|^2/2 : y \in \Lambda^*\setminus\{0\}\}\big).
\end{equation}
The existence of a lattice $\Lambda$ that makes $\Delta_1 \ge (1+o(1))c/(4\pi
e)$ follows from an averaging argument using the Siegel mean value theorem;
in fact, $\Lambda$ can even be chosen to be a self-dual integral lattice (see
Theorem~9.5 in \cite[Chapter~II]{MH}).
\end{proof}

For comparison, the lattices in Table~\ref{table:narain} with $2 \le c \le 8$
cannot be isometric to lattices of the form $\Lambda \times \Lambda^*$,
because $\Delta_1$ is too large: one of $\Lambda$ or $\Lambda^*$ would
violate the linear programming bound for sphere packing in $\R^c$. The only
way to circumvent this obstacle is to use a nonzero antisymmetric matrix $M$,
and the averaging argument in Section~\ref{ss:spectralgap} takes advantage of
$M$ as well as $\Lambda$.

The Coxeter-Todd and Barnes-Wall lattices can be obtained through
Proposition~\ref{prop:narainequiv}, but the prettiest constructions we have
found use a variant of this construction: the lattice
\begin{equation}
T\{ (u + Mv, \, v) : u \in \Lambda, \,v \in \Lambda^*\}
\end{equation}
is a Narain lattice if and only if $\langle Mv, v \rangle \in \Z$ for all $v
\in \Lambda^*$. This equivalence follows immediately from the formula
$[T(x,y),T(x,y)] = 2\langle x,y \rangle$. If $M$ is antisymmetric, then
$\langle Mv,v\rangle=0$ automatically, while otherwise it is a matter of
compatibility between $M$ and $\Lambda^*$. If $\Lambda^*$ is a rescaling of
an integral lattice, then taking $M$ to be a corresponding multiple of $I$
works, and we can of course add to it any antisymmetric matrix.

To obtain the Barnes-Wall lattice, we start with the $E_8$ root lattice,
which is an even unimodular lattice in $\R^8$. It has the structure of a
module over the Gaussian integers $\Z[i]$; in other words, there exists $J
\in \textup{O}(8)$ such that $J^2=-I$ and multiplication by $J$ preserves $E_8$. If we
let $M = (I+J)/\sqrt{2}$, then
\begin{equation}
T\{(u + Mv, \, v) : u \in 2^{1/4} E_8, \,v \in 2^{-1/4}E_8\}
\end{equation}
is a Narain lattice, and one can check that it is isometric to the
Barnes-Wall lattice (rescaled to have determinant $1$). One can compute
$\Delta_1$ as follows. If we set $u = 2^{1/4} x$ and $v = 2^{-1/4} y$ with
$x,y \in E_8$, then checking that $\Delta_1 = \sqrt{2}$ amounts to showing
that
\begin{equation}
\left| x + \frac{I+J}{2}y \right|^2 + \frac{1}{2}|y|^2 \ge  2
\end{equation}
unless $x=y=0$. If $y=0$ or $|y|^2 \ge 4$, then the inequality trivially
holds, and therefore the interesting case is $|y|^2=2$. In that case,
$(I+J)y$ is a vector of norm $4$ in $E_8$ since $|(I+J)y|^2 =
|y|^2+|Jy|^2=4$, and therefore $(I+J)y/2$ is a deep hole of $E_8$ (see
\cite[p.~121]{SPLAG}), which is at distance $1$ from the nearest points of
$E_8$.

Similarly, the Coxeter-Todd lattice (again rescaled to have determinant $1$)
is given by
\begin{equation}
T\{(u + \sqrt{3} v, \, v) : u \in (4/3)^{1/4} E_6, \,v \in (4/3)^{-1/4} E_6\},
\end{equation}
with no need for an antisymmetric matrix. The remaining case is $c=7$, where
we do not know of a previous occurrence of the best lattice we have found.
It achieves $\Delta_1 = \sqrt{4/3}$ by using
\begin{equation} \begin{split}
\Lambda &= (2/3^{1/4}) D_7^*\\
&= (2/3^{1/4}) \big(\Z^7 \cup \big(\Z^7 + (1/2,1/2,\dotsc,1/2)\big)\big)
\end{split} \end{equation}
and the antisymmetric matrix
\begin{equation}
M = \frac{1}{\sqrt{3}}\begin{pmatrix}
\phantom{-}0 & \phantom{-}1 & \phantom{-}1 & \phantom{-}1 & -1 & \phantom{-}1 & \phantom{-}1\\
-1 & \phantom{-}0 & \phantom{-}1 & -1 & \phantom{-}1 & \phantom{-}1 & \phantom{-}1\\
-1 & -1 & \phantom{-}0 & \phantom{-}1 & \phantom{-}1 & \phantom{-}1 & -1\\
-1 & \phantom{-}1 & -1 & \phantom{-}0 & \phantom{-}1 & -1 & \phantom{-}1\\
\phantom{-}1 & -1 & -1 & -1 & \phantom{-}0 & \phantom{-}1 & \phantom{-}1\\
-1 & -1 & -1 & \phantom{-}1 & -1 & \phantom{-}0 & \phantom{-}1\\
-1 & -1 & \phantom{-}1 & -1 & -1 & -1 & \phantom{-}0
\end{pmatrix}.
\end{equation}

In the above constructions, we built the Coxeter-Todd and Barnes-Wall
lattices using rescalings of $E_6$ and $E_8$, respectively, but taking $M$ to
be a linear combination of the identity matrix and an antisymmetric matrix. In
fact, the use of the identity matrix is unnecessary: one can use exactly the
same $c$-dimensional lattices, and replace $M$ with an antisymmetric matrix.

\section{The Hardy-Littlewood circle method} \label{app:circle}

The remaining circle method calculations work as follows, in the notation of
Section~\ref{ss:counting}. Recall that we are trying to approximate the
integrand
\begin{equation}
 \sum_{(x,y) \in B_r \cap (\Z^c)^2} e^{2\pi i (x \cdot y - \ell) w}
\end{equation}
and integrate it over the major arcs, which consist of the $w$ satisfying
\begin{equation}
 \left|w-\frac{a}{b}\right| \le \frac{1}{r^{2-\varepsilon}}
\end{equation}
for rationals $a/b$ with $1 \le b \le r^\varepsilon$.

If $w = a/b+u$ with $u$ small, we can decompose our sum into residue classes
modulo $b$ and write the integrand as
\begin{equation}
\sum_{(x,y) \in B_r \cap (\Z^c)^2} e^{2\pi i (x \cdot y - \ell) w} = \sum_{(\bar{x},\bar{y}) \in (\Z/b\Z)^{2c}} e^{2\pi i (\bar{x} \cdot \bar{y} - \ell) a/b} \sum_{\substack{(x,y) \in B_r \cap (\Z^c)^2 \\ (x,y) \equiv (\bar{x},\bar{y}) \pmod{b}}} e^{2\pi i (x \cdot y- \ell) u}.
\end{equation}
(Here we use the fact that $e^{2\pi i (x \cdot y - \ell) a/b}$ depends only
on $x \cdot y$ modulo $b$.) Because $u$ is small, we can approximate the last
sum by an integral, to obtain
\begin{equation}
\sum_{\substack{(x,y) \in B_r \cap (\Z^c)^2 \\ (x,y) \equiv (\bar{x},\bar{y}) \pmod{b}}} e^{2\pi i (x \cdot y- \ell) u}  \sim \frac{1}{b^{2c}} \int_{(x,y) \in B_r} dx\, dy\, e^{2\pi i (x \cdot y- \ell) u}.
\end{equation}
If we let
\begin{equation}
S(a,b) =  \sum_{(\bar{x},\bar{y}) \in (\Z/b\Z)^{2c}} e^{2\pi i (\bar{x} \cdot \bar{y} - \ell) a/b},
\end{equation}
then the integral over the major arc at $a/b$ is asymptotic to
\begin{equation}
\frac{S(a,b)}{b^{2c}} \int_{|u| \le r^{\varepsilon-2}} du  \int_{(x,y) \in B_r} dx\, dy\, e^{2\pi i (x \cdot y- \ell) u}.
\end{equation}
As $r \to \infty$, replacing $u$ with $ur^2$ and $(x,y)$ with $(x,y)/r$ yields
\begin{equation}
\frac{S(a,b)}{b^{2c}} r^{2c-2} \int_{\R} du  \int_{(x,y) \in B_1} dx\, dy\, e^{2\pi i (x \cdot y) u}.
\end{equation}
Aside from justifying the quality of these approximations,\footnote{Recall
that the result is not even true when $c \le 2$.} we have shown that
$\#\{(x,y) \in B_r \cap (\Z^c)^2 : x \cdot y = \ell\} $ is asymptotic to
\begin{equation}
\sum_{b \ge 1} \sum_{\substack{1 \le a \le b\\ \gcd(a,b)=1}} \frac{S(a,b)}{b^{2c}} \sigma_\infty(B_r)
\end{equation}
as $r \to \infty$.

Let
\begin{equation}
S(b) = \sum_{\substack{1 \le a \le b\\ \gcd(a,b)=1}}S(a,b).
\end{equation}
All that remains is to justify that
\begin{equation}
\sum_{b \ge 1} \frac{S(b)}{b^{2c}} = \prod_{\text{$p$ prime}} \sigma_p,
\end{equation}
where
\begin{equation}
\sigma_p = \lim_{n \to \infty} \frac{\#\{(\bar{x},\bar{y}) \in \big((\Z/p^n\Z)^{c}\big)^2 : \bar{x}\cdot \bar{y} \equiv  \ell \pmod{p^n}\}}{p^{(2c-1)n}}.
\end{equation}
By the Chinese remainder theorem, $S(bb')=S(b)S(b')$ whenever $\gcd(b,b')=1$.
Thus, factoring into prime powers reduces what we need to prove to the case
\begin{equation}
\sum_{k \ge 0} \frac{S(p^k)}{p^{2kc}} = \sigma_p
\end{equation}
for $p$ prime. To obtain this formula, we write the partial sums as
\begin{equation} \begin{split}
\sum_{k = 0}^n \frac{S(p^k)}{p^{2kc}} &= \sum_{k=0}^n \sum_{\substack{1 \le a \le p^k\\ \gcd(a,p^k)=1}} \frac{1}{p^{2kc}} \sum_{(\bar{x},\bar{y}) \in (\Z/p^k\Z)^{2c}} e^{2\pi i (\bar{x} \cdot \bar{y} - \ell) a/p^k}\\
&= \sum_{k=0}^n \sum_{\substack{1 \le a \le p^k\\ \gcd(a,p^k)=1}} \frac{1}{p^{2nc}} \sum_{(\bar{x},\bar{y}) \in (\Z/p^n\Z)^{2c}} e^{2\pi i (\bar{x} \cdot \bar{y} - \ell) a/p^k}
\end{split} \end{equation}
and then set $a'=ap^{n-k}$ and use the identity
\begin{equation}
\sum_{a=1}^b e^{2\pi i ma/b} = \begin{cases}
b & \text{if $m$ is a multiple of $b$, and}\\
0 & \text{otherwise}
\end{cases}
\end{equation}
to conclude that
\begin{equation} \begin{split}
\sum_{k = 0}^n \frac{S(p^k)}{p^{2kc}} &=\frac{1}{p^{2nc}} \sum_{a'=1}^{p^n}  \sum_{(\bar{x},\bar{y}) \in (\Z/p^n\Z)^{2c}} e^{2\pi i (\bar{x} \cdot \bar{y} - \ell) a'/p^n}\\
&= \frac{\#\{(\bar{x},\bar{y}) \in \big((\Z/p^n\Z)^{c}\big)^2 : \bar{x} \cdot \bar{y} \equiv \ell \pmod{p^n}\}}{p^{(2c-1)n}} ,
\end{split} \end{equation}
as desired.

\section{Counting solutions modulo prime powers} \label{app:counting}

Let $p$ be a prime and $\ell$ be any integer, and let
\begin{equation}
V(p^n,\ell) = \{(x,y) \in \big((\Z/p^n\Z)^{c}\big)^2 : x\cdot y \equiv  \ell \pmod{p^n}\}.
\end{equation}
In this appendix we compute
\begin{equation}
\sigma_p = \lim_{n \to \infty} \frac{\# V(p^n,\ell)}{p^{(2c-1)n}}.
\end{equation}
Note that this scaling is sensible, since we are putting one constraint
modulo $p^n$ on $2c$ variables, but the constant of proportionality
$\sigma_p$ will depend on $p$ and $\ell$.

We will show that $\# V(p^n,\ell)$ satisfies the recurrence relation
\begin{equation}
\# V(p^n,\ell)  = \big(p^{cn} - p^{c(n-1)}\big)p^{(c-1)n} + p^c \#V(p^{n-1},\ell/p),
\end{equation}
where we interpret $\# V(p^{n-1},\ell/p)$ as $0$ if $\ell$ is not divisible
by $p$. Once we have the recurrence, we find that
\begin{equation}
\frac{\# V(p^n,\ell)}{p^{(2c-1)n}}  = 1-p^{-c} + p^{-(c-1)} \frac{\#V(p^{n-1},\ell/p)}{p^{(2c-1)(n-1)}},
\end{equation}
and it follows immediately that
\begin{equation} \begin{split}
\sigma_p &= (1-p^{-c})\sum_{i=0}^k p^{-(c-1)i} \\
&= \frac{(1-p^{-c})(1-p^{-(c-1)(k+1)})}{1-p^{-(c-1)}}
\end{split} \end{equation}
when $\ell$ is divisible by $p^k$ but no higher power of $p$, where we take
$k=\infty$ and $p^{-(c-1)(k+1)}=0$ if $\ell=0$.

To prove the recurrence, we divide into two cases. Suppose first that $x
\not\equiv 0 \pmod{p}$. Every integer not divisible by $p$ is a unit modulo
$p^n$ (i.e., it has a multiplicative inverse modulo $p^n$), and so some
coordinate of $x$ is a unit, say $x_i$. Then we can choose the other
coordinates $y_j$ of $y$ arbitrarily, and achieve $x \cdot y \equiv \ell
\pmod{p^n}$ through a unique choice for $y_i$, namely
\begin{equation}
y_i \equiv x_i^{-1} \left(\ell - \sum_{j \ne i} x_j y_j\right) \pmod{p^n}.
\end{equation}
There are $p^{cn} - p^{c(n-1)}$ choices of $x$ that are not divisible by $p$,
and each of them has $p^{(c-1)n}$ choices of $y$ that work with it. Thus,
there are $\big(p^{cn} - p^{c(n-1)}\big)p^{(c-1)n}$ solutions to $x \cdot y
\equiv \ell \pmod{p^n}$ with $x \not\equiv 0 \pmod{p}$.

The remaining case is when $x \equiv 0 \pmod{p}$. In that case, let $x = p
x'$, where $x'$ is defined modulo $p^{n-1}$. The only way we can have $x
\cdot y \equiv \ell \pmod{p^n}$ is if $\ell$ is divisible by $p$. If so, for
each $y'$ modulo $p^{n-1}$ satisfying
\begin{equation}
x' \cdot y' \equiv \ell/p \pmod{p^{n-1}},
\end{equation}
there are $p^c$ choices of $y$ modulo $p^n$ that reduce to $y'$ modulo
$p^{n-1}$ (namely, $y' + p^{n-1}z$ for any vector $z$ modulo $p$), and there
are therefore $p^c \# V(p^{n-1},\ell/p)$ solutions with $x \equiv 0
\pmod{p}$. Thus, the recurrence relation holds.

\end{spacing}

\providecommand{\href}[2]{#2}\begingroup\raggedright\endgroup

\end{document}